\documentclass[11pt,letterpaper]{article}
\usepackage{graphicx}
\usepackage{epsfig}
\usepackage{amssymb}
\usepackage{amsthm}
\usepackage{amsfonts}
\usepackage{amsmath}
\usepackage{multicol,multirow}
\usepackage[margin=1in]{geometry}
\usepackage{bbm} 

\usepackage{enumitem}
\usepackage{subcaption}
\usepackage{hhline}

\newtheorem{theorem}{Theorem}
\newtheorem{lemma}[theorem]{Lemma}
\newtheorem{corollary}[theorem]{Corollary}

\newtheorem{definition}{Definition}

\newtheorem{remark}{Remark}

\newcommand{\trace}{\mathrm{Tr}}

\newcommand{\R}{\mathbb{R}}
\newcommand{\eig}{\mathcal{E}}

\newcommand{\G}{\mathcal{G}}
\newcommand{\pomega}{\mathcal{P}_{\Omega}}
\newcommand{\pomegac}{\mathcal{P}_{\Omega^c}}
\newcommand{\operator}{\mathcal{L}^{\alpha}}
\newcommand{\projop}{ \mathcal{P}}
\newcommand{\ged}{\textrm{GED}}
\newcommand{\bones}{\mathbf{1}}

\newcommand{\norm}[1]{\left\lVert#1\right\rVert} 

\DeclareMathOperator{\conv}{conv}
\DeclareMathOperator{\sign}{sign}

\title{Convex Graph Invariant Relaxations for Graph Edit Distance} 
\author{Utkan Onur Candogan$^\dag$ and Venkat Chandrasekaran$^{\dag,\ddag}$ \thanks{Email: utkan@caltech.edu, venkatc@caltech.edu.  The authors were supported in part by NSF grants CCF-1350590 and CCF-1637598, by AFOSR grant FA9550-16-1-0210, and by a Sloan research fellowship.} \vspace{0.25in} \\ $^\dag$ Department of Electrical Engineering\\ $^\ddag$ Department of Computing and Mathematical Sciences \\ California Institute of Technology \\ Pasadena, CA 91125}

\begin{document}

\maketitle

\begin{abstract}
The edit distance between two graphs is a widely used measure of similarity that evaluates the smallest number of vertex and edge deletions/insertions required to transform one graph to another.  It is NP-hard to compute in general, and a large number of heuristics have been proposed for approximating this quantity.  With few exceptions, these methods generally provide upper bounds on the edit distance between two graphs.  In this paper, we propose a new family of computationally tractable convex relaxations for obtaining lower bounds on graph edit distance.  These relaxations can be tailored to the structural properties of the particular graphs via \emph{convex graph invariants}.  Specific examples that we highlight in this paper include constraints on the graph spectrum as well as (tractable approximations of) the stability number and the maximum-cut values of graphs.  We prove under suitable conditions that our relaxations are tight (i.e., exactly compute the graph edit distance) when one of the graphs consists of few eigenvalues.  We also validate the utility of our framework on synthetic problems as well as real applications involving molecular structure comparison problems in chemistry.
\end{abstract}


\section{Introduction}\label{Introduction}

Graphs are widely used to represent the structure underlying a collection of interacting entities.  A common computational question arising in many contexts is that of measuring the similarity between two graphs.  For example, the unknown functions of biological structures such as proteins, RNAs and genes are often deduced from structures which have similar sequences with known functions \cite{ibragimov2013gedevo, jiang2002general, memivsevic2012c, sharan2005conserved,sharan2006modeling}.  Evaluating graph similarity also plays a central role in various pattern recognition applications \cite{conte2004thirty,neuhaus2006edit}, specifically in areas such as handwriting recognition \cite{fischer2013fast,lu1991hierarchical}, fingerprint classification \cite{isenor1986fingerprint,neuhaus2005graph} and face recognition \cite{wiskott1997face}.


\begin{figure}[hbt]
\begin{center}
\includegraphics[scale=0.25]{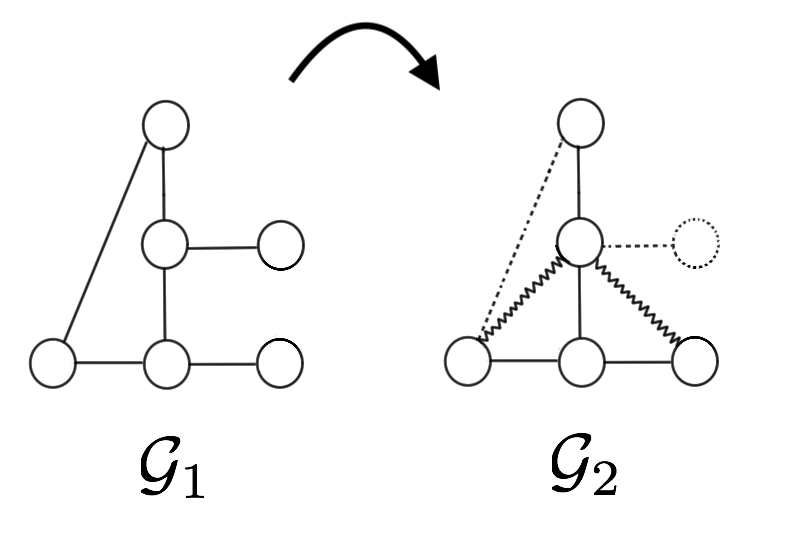}
\caption{ An instance of a graph edit distance problem in which we wish to calculate minimum number of edit operations required for transforming graph $\G_1$ to graph $\G_2$. The edit operations are encoded by line style: The dashed graph elements are to be removed from $\G_1$ and the zigzagged graph elements are to be added to $\G_1$ for transforming $\G_1$ to $\G_2$. Assuming a cost of $1$ for every edit operation, we conclude that the graph edit distance between $\G_1$ and $\G_2$ is 5. }\label{FigureIntroGED}
\end{center}
\end{figure}


The notion of similarity that is the most commonly considered is the \emph{graph edit distance} \cite{sanfeliu1983distance}.  The edit distance $\ged(\G_1,\G_2)$ between two graphs $\G_1$ and $\G_2$ is the smallest number of operations required to transform $\G_1$ into $\G_2$ by a sequence of edits or changes applied to the vertices and edges of $\G_1$.  A particular sequence of edit operations transforming $\G_1$ into $\G_2$ is usually referred to as an edit path. For unlabeled graphs, the permissible set of edit operations are usually insertions/deletions of vertices/edges.  For labeled graphs, the set of permissible edit operations can also include vertex/edge relabelings.  In some situations, certain types of edits are considered more `severe' than others and different edits have different costs associated to them; in such cases, the edit distance is the smallest cost over all edit paths that transform one graph to another, where the cost of an edit path is the sum of the costs of the edits that compose the path.  See Figure \ref{FigureIntroGED} for an illustration of a simple graph edit distance problem.

The problem of computing the graph edit distance is NP-hard in general \cite{garey2002computers}, and in practice exact calculation of the edit distance is only feasible for small-sized graphs.  Thus, significant efforts have been directed towards developing computationally tractable heuristics for approximating the edit distance \cite{abu2017graph,bougleux2017graph,daller2018approximate, justice2006binary, leordeanu2009integer, liu2014gnccp, riesen2009approximate, riesen2014computing, zeng2009comparing} or for exactly computing the edit distance for graphs from specific families such as planar graphs \cite{neuhaus2004error}.  These methods are largely combinatorial in nature, and most of them aim at identifying an edit path that transforms one graph to the other.  Consequently, much of the prior literature on this topic provides techniques that lead to upper bounds on the edit distance between two graphs.  In contrast, far fewer approaches have been proposed for obtaining lower bounds on the edit distance.  We are aware of three notable examples; in \cite{justice2006binary} and \cite{riesen2014computing} the authors propose tractable linear programming relaxations of intractable combinatorial optimization formulations of the edit distance based on $0/1$ optimization and the quadratic assignment problem, respectively, while in  \cite{zeng2009comparing} the authors propose a combinatorial algorithm based on a multiset representation of the graph that enables the efficient computation of upper and lower bounds of the graph edit distance.
 
In this paper, we develop a computationally efficient framework for obtaining lower bounds on the graph edit distance.  Our contributions differ qualitatively from the prior literature in two aspects.  First, our approach can be tailored to the structural properties of the specific graphs at hand based on the notion of a \emph{convex graph invariant}.  These lead to useful lower bounds on the edit distance if one of the graphs is `suitably structured'.  Second, we give a theoretical analysis of conditions on pairs of graphs under which a tractable semidefinite relaxation based on the spectral properties of a graph provably computes the edit distance, i.e., our lower bound is tight.

%

\subsection{Our Framework}

Much of the focus of our development and our analysis is on the edit distance between two unlabeled graphs on the same number of vertices, with edge deletions and insertions being the allowed edit operations.  We discuss in Section~\ref{SectionExtendedFramework} an extension of our framework to settings in which the number of vertices in the two graphs may be different and in which vertex deletions and insertions are also allowed.  Let $A_1\in\mathbb{S}^n$ and $A_2\in\mathbb{S}^n$ represent the adjacency matrices of two unweighted, unlabeled, simple and loopless graphs $\G_1$ and $\G_2$ on $n$ vertices.  (Here $\mathbb{S}^n$ denotes the space of $n \times n$ real symmetric matrices.)  The following optimization problem gives a combinatorial formulation of the computation of the edit distance between $\G_1$ and $\G_2$:
\begin{equation}\label{Optimization Combinatorial}
\begin{aligned}
\ged(\G_1,\G_2)   &= \min_{\substack{X,E\in\mathbb{S}^n}}~ \sum_{1\leq i < j \leq n} \mathbbm{1}_{E_{ij} \neq 0}\\
&~~~~~~s.t.~~~~X+E = A_2\\
&~~~~~~~~~~ ~~~~X \in \{ \Pi A_1 \Pi^T ~:~ \Pi \text{ is an $n\times n$ permutation matrix.} \}\\
&~~~~~~~~~~ ~~~~ E_{ij}\, \in \,\{-1,0,1\} ~ \forall\, i,j \in \{1,\dots,n\}.
\end{aligned}
\end{equation}
The function $\mathbbm{1}_{\{\cdot\}}$ denotes the usual indicator function, the decision variable $X$ is an adjacency matrix representing $\G_1$, and the decision variable $E$ specifies the edge deletions and insertions made to $\G_1$ to obtain $\G_2$.  One aspect of the problem \eqref{Optimization Combinatorial} is that its formulation is not symmetric with respect to $\G_1$ and $\G_2$, although the optimal value remains unchanged if $A_1$ and $A_2$ are swapped in the problem description, i.e., $\ged(\G_1,\G_2) = \ged(\G_2,\G_1)$; we revisit this point in the sequel.  Unsurprisingly, solving (\ref{Optimization Combinatorial}) is intractable in general as calculating the graph edit distance is an NP-hard problem.  Our approach in this paper is to obtain tractable convex relaxations of (\ref{Optimization Combinatorial}).  Relaxing the objective to a convex function is a straightforward matter; specifically, as the absolute value function constitutes a lower bound on the indicator function in the range $[-1,1]$; we replace the objective of (\ref{Optimization Combinatorial}) with the convex function $\tfrac{1}{2} \|E\|_{\ell_1}$, where $\|\cdot\|_{\ell_1}$ denotes the (entrywise) $\ell_1$ norm that sums the absolute values of the entries of a matrix.

The main remaining source of difficulty with obtaining a convex relaxation of \eqref{Optimization Combinatorial} is to identify a tractable convex approximation of a set of the form $\{ \Pi A \Pi^T \,:\, \Pi \text{ is an } n\times n \text{ permutation matrix}\}$ for a given matrix $A \in \mathbb{S}^n$.  When $A$ specifies an adjacency matrix of a graph, this set consists of all the adjacency matrices that describe the graph, thus highlighting the \emph{structural} attributes of the graph that remain invariant to vertex relabeling.  Consequently, we seek a convex approximation that similarly remains invariant to vertex relabeling.  We describe next a notion from \cite{chandrasekaran2012convex} that aims to address this challenge:

\begin{definition}\cite{chandrasekaran2012convex} A set $\mathcal{C} \subset \mathbb{S}^n$ is an \emph{invariant convex set} if it is convex and if $M \in \mathcal{C}$ implies that $\Pi M \Pi^T \in \mathcal{C}$ for all $n \times n$ permutation matrices $\Pi$.
\end{definition}

Invariant convex sets provide a useful convex modeling framework to constrain graph properties that remain invariant to vertex relabeling.  In particular, suppose that $\mathcal{C}_{\G_1} \subset \mathbb{S}^n$ is an invariant convex set that contains $\{ \Pi A_1 \Pi^T \,:\, \Pi \text{ is an } n\times n \text{ permutation matrix}\}$ and has an efficient description. Then, the following \emph{convex program} provides a lower bound on $\ged(\G_1,\G_2)$:
\begin{equation}\tag{$P$}\label{Optimization P}
\begin{aligned}
\ged_{LB}(\G_1,\G_2; \mathcal{C}_{\G_1})   &= \min_{\substack{X,E\in\mathbb{S}^n}}~ \tfrac{1}{2} \norm{E}_1 \\
&~~~~~~s.t.~~~~X+E = A_2\\
&~~~~~~~~~~ ~~~~X \in \mathcal{C}_{\G_1},\\
\end{aligned}
\end{equation}
It is evident that this problem provides a lower bound on $\ged(\G_1,\G_2)$ as the objective function here is a lower bound of the objective of \eqref{Optimization Combinatorial} over the constraint set of \eqref{Optimization Combinatorial}, and further the constraint set of \eqref{Optimization P} is an outer approximation of the constraint set of \eqref{Optimization Combinatorial}.  Unlike the optimal value $\ged(\G_1,\G_2)$ of \eqref{Optimization Combinatorial}, the optimal value $\ged_{LB}(\G_1,\G_2; \mathcal{C}_{\G_1})$ of \eqref{Optimization P} is not symmetric; specifically, if $\mathcal{C}_{\G_2}$ is some invariant convex set containing $\{ \Pi A_2 \Pi^T \,:\, \Pi \text{ is an } n\times n \text{ permutation matrix}\}$, then in general $\ged_{LB}(\G_1,\G_2; \mathcal{C}_{\G_1}) \neq \ged_{LB}(\G_2,\G_1; \mathcal{C}_{\G_2})$.  Therefore, in practice we propose computing both $\ged_{LB}(\G_1,\G_2; \mathcal{C}_{\G_1})$ and $\ged_{LB}(\G_2,\G_1; \mathcal{C}_{\G_2})$ for some invariant convex sets $\mathcal{C}_{\G_1}$ and $\mathcal{C}_{\G_2}$ corresponding to $\G_1$ and $\G_2$ respectively, and taking the larger of these quantities as both constitute lower bounds on $\ged(\G_1,\G_2)$.

This discussion leads naturally to the following question -- which invariant convex set $\mathcal{C}_\G$ best captures the structural properties of a graph $\G$?  Employing such a set in the relaxation \eqref{Optimization P} would provide better lower bounds on the edit distance.  Letting $A \in \mathbb{S}^n$ be an adjacency matrix of $\G$, the `tightest' invariant convex set that contains the collection $\{ \Pi A \Pi^T \,:\, \Pi \text{ is an } n\times n \text{ permutation matrix}\}$ is simply the convex hull of this collection.  However, this convex hull is intractable to describe for general graphs $\G$ (unless P=NP).  As a result, it is of interest to obtain computationally tractable invariant convex relaxations that reflect the structure in $\G$.  In the next subsection, we give a list of invariant convex sets that are tractable to compute and that can be `tuned' to the structure of $\G$.  These invariants can either be used individually or combined (at increased computational expense), thus yielding a flexible and powerful framework for obtaining bounds on the graph edit distance.  The focus of the rest of the paper is on investigating the utility of these invariants theoretically as well as via numerical experiments; we give a summary of our main contributions in Section \ref{SectionContributions}.

\subsection{Convexity and Graph Invariants} \label{SectionConvexInvariantSets}
We list here a few examples of invariant convex sets that play a prominent role in this paper; we refer the interested reader to \cite{chandrasekaran2012convex} for a more exhaustive list as well as additional properties of invariant convex sets.

\textbf{Loopless and edge weight constraints} Looplessness and edge weight bounds are not especially powerful constraints, but they nonetheless serve as simple examples of invariant convex sets.  Looplessness corresponds to the constraint set $\{ M \in \mathbb{S}^n ~|~  M_{ii} = 0 \text{ for }  i=1,\dots,n\}$, and bounds on the edge weights for unweighted graphs (for example) can be specified via the set $\{ M \in \mathbb{S}^n ~|~  0 \leq M_{ij} \leq 1 \text{ for }  i,j=1,\dots,n\}$.

\textbf{Spectral invariants} Let $\G$ be a graph represented by an adjacency matrix $A \in \mathbb{S}^n$ with eigenvalues $\lambda(A) \in \mathbb{R}^n$.  The smallest convex set containing all graphs that are isospectral to $\G$ is given by the following \emph{Schur-Horn orbitope} associated to $A$ \cite{sanyal2011orbitopes}:
\begin{align*}
\mathcal{C}_{\mathcal{SH}(\G)}  &= \conv\{M \in \mathbb{S}^n~|~ \lambda(M) = \lambda(A)\}.
\end{align*}
This set consists precisely of those matrices whose spectra are majorized by $\lambda(A)$.  One can replace the list of eigenvalues in this example with the degree sequence of a graph, and in a similar vein, consider the convex hull of all adjacency matrices representing graphs with the same degree sequence; see \cite{chandrasekaran2012convex} for more details.

A prominent way in which invariant convex sets can be constructed is via sublevel sets of convex graph invariants:

\begin{definition}  \cite{chandrasekaran2012convex} A function $f:\mathbb{S}^n \rightarrow \mathbb{R}$ is a \emph{convex graph invariant} if it is convex and if $f(M) = f(\Pi M \Pi^T)$ for all $M \in \mathbb{S}^n$ and all $n \times n$ permutation matrices $\Pi$.
\end{definition}

As with invariant convex sets, convex graph invariants characterize structural properties of a graph that are invariant to vertex relabeling.  The following convex graph invariants play a prominent role in our work:

\textbf{Inverse of the stability number} A stable set (also known as independent set) of a graph $\G$ is a subset of vertices of $\G$ such that no two vertices in the subset are connected. The stability number of a graph $\G$ is a graph invariant that is equal to the size of the largest stable set of $\G$.  It was shown by Motzkin and Straus \cite{motzkin1965maxima} that the inverse of the stability number admits the following variational description, where $A \in \mathbb{S}^n$ is an adjacency matrix representing $\G$:
\begin{align*}
\textrm{inv-stab-number}(A) &= \min_{x\in\mathbb{R}^n} x'(I+A)x\\
&~~~~~s.t. ~~~\sum_i x_i =1, ~ x_i \geq 0 \text{ for } i=1,\dots,n.
\end{align*}
As the stability number of a graph is NP-hard to compute for general graphs (the above program may be reformulated as a conic program with respect to the completely positive cone), the following tractable relaxation based on doubly nonnegative matrices is widely employed:
\begin{equation} \label{EqnInvStRelaxation}
\begin{aligned}
f(A) &= \min_{X\in\mathbb{S}^n} \trace(X(I+A) )\\
&~~~~~s.t. ~~~X\succeq 0,~ \bones' X \bones = 1, X_{ij} \geq 0 \text{ for } i,j=1,\dots,n.
\end{aligned}
\end{equation}
One can check that both $\textrm{inv-stab-number}(A)$ and $f(A)$ are concave graph invariants.

\textbf{Maximum cut} The maximum cut value of a graph is the maximum over all partitions of the vertices of the sum of the weights of the edges between the partitions.  For a graph $\G$ specified by adjacency matrix $A \in \mathbb{S}^n$, the maximum cut value is given as:
\begin{align*}
\textrm{max-cut}(A) &= \max_{y\in\{-1,1\}^n} \tfrac{1}{4}\,\sum_{i,j} A_{i,j} (1-y_iy_j).
\end{align*}
As this value is NP-hard to compute for general graphs, the following celebrated efficiently-computable relaxation is commonly used \cite{goemans1995improved}:
\begin{equation}\label{EqnMaxCutRelaxation}
\begin{aligned}
g(A) &= \max_{X\in\mathbb{S}^n} \tfrac{1}{4}\, \trace( A\, (\bones \bones^T-X ) )\\
&~~~~~s.t. ~~~X \succeq 0, ~ X_{ii}=1 \text{ for } i=1,\dots,n.
\end{aligned}
\end{equation}
Both $\textrm{max-cut}(A)$ and $g(A)$ are convex graph invariants as they are each invariant under conjugation of the argument by a permutation matrix and they are each expressed as a pointwise maximum of affine functions.


\subsection{Our Contributions} \label{SectionContributions}
The invariant convex sets listed in the previous section when used in the context of the optimization problem \eqref{Optimization P} all lead to valid lower bounds on the edit distance between two graphs.  The question then is whether certain invariants are more naturally suited to particular structural properties of graphs.  The main focus of this paper is on identifying attributes of graphs for which the invariants described above are well-suited, and evaluating in these contexts the quality of the bounds obtained via \eqref{Optimization P} both theoretically and through numerical experiments.  Specifically, we say that an invariant convex constraint set $\mathcal{C}_\G$ is well-suited to the structure of a graph $\G$ if $\ged_{LB}(\G,\G';\mathcal{C}_\G)$ provides a tight (or high-quality) lower bound of $\ged(\G,\G')$ for all graphs $\G'$ that are obtained via a small number of edge deletions and insertions applied to $\G$ (here `small' is interpreted relative to the total number of edges in $\G$).

In Section \ref{SectionSchurHorn}, we investigate theoretically the effectiveness of the Schur-Horn orbitope as an invariant convex set in providing lower bounds on the graph edit distance via \eqref{Optimization P}.  We consider a stylized setting in which a graph $\G$ on $n$ vertices is modified to a graph $\G'$ by adding or removing at most $d$ edges incident to each vertex of $\G$.  We prove in Theorem~\ref{TheoremMainTheorem} (see Section \ref{SectionMainResult}) that the optimal value of the convex program \eqref{Optimization P} with a Schur-Horn orbitope constraint set equals the graph edit distance between $\G$ and $\G'$, i.e., $\ged_{LB}(\G,\G';\mathcal{C}_{SH(\G)}) = \ged(\G,\G')$ provided: 1) $d$ is sufficiently small, 2) $\G$ has eigenspaces with the property that there exists a linear combination of the associated projection operators onto these spaces so that the largest entry in magnitude is suitably bounded, and 3) any matrix supported on entries corresponding to edits of $\G$ has only a small amount of its energy on each of the eigenspaces of $\G$; see Theorem~\ref{TheoremMainTheorem} for precise details.  Conditions similar to the third requirement appear in the authors' earlier work on employing the Schur-Horn orbitope in the context of the planted subgraph problem \cite{candogan2018finding}.  However, the second condition is novel and is motivated by the context of the present paper on graph edit distance.  Under the additional assumption that $\G$ is vertex-transitive, Corollary \ref{CorollaryMain} provides a simple formula on the maximum allowable number $d$ of edge additions/deletions per vertex of $\G$; we illustrate the utility of this formula by computing bounds on $d$ for many graph families such as Johnson graphs, Kneser graphs, Hamming graphs and other strongly regular graphs.  Indeed, for some of these families, our results are `order-optimal' in the sense that our bounds on $d$ are on the order of the degree of $\G$. The proofs of the main results of Section \ref{SectionSchurHorn} are given in Section \ref{SHProofs}.

In Section \ref{SectionExploratory}, we conduct a detailed numerical evaluation of the power and limitations of convex invariant relaxations based on the inverse stability number (via the tractable approximation (\ref{EqnInvStRelaxation})) and the maximum cut value (via the tractable approximation (\ref{EqnMaxCutRelaxation})).  We do not provide precise theoretical guarantees due to a lack of a detailed characterization of the facial structure of the associated convex sets.  Nonetheless, we identify classes of graph edit distance problems for which these constraints produce high-quality lower bounds.  Specifically, we observe that a convex relaxation based on the Motzkin-Straus approximation of the inverse of the stability number provides useful lower bounds on graph edit distance if one of the graphs has the property that the removal of any edge increases the graph's stability number; graphs with such a property have been studied in the extremal graph theory literature \cite{jou2000number} and we refer the reader to Section \ref{SectionStabilityNumber} for further details.  Similarly, in Section \ref{SectionMaxcut}, we observe that a convex relaxation based on the Goemans-Williamson approximation of the maximum cut value produces effective lower bounds on the graph edit distance if the addition of any edge to one of the graphs increases that graph's maximum cut value; windmill graphs are a prominent example that possess such a property. In both Sections \ref{SectionStabilityNumber} and \ref{SectionMaxcut}, we present  empirical results that corroborate our observations.

In Section \ref{SectionRealExperiment} we demonstrate the utility of our framework in providing lower bounds on the average pairwise graph edit distance in two chemistry datasets consisting of a collection of molecules known as Alkanes and Polycyclic Aromatic Hydrocarbons (PAH).  The PAH dataset in particular consists of large structures for which exact computation of graph edit distance is prohibitively expensive.  The best-known upper bound on the average graph edit distance over all pairs of graphs in this dataset is 29.8, and to the best of our knowledge, the exact value of this quantity is not known \cite{bougleux2017graph}. Indeed, much of the literature featuring the PAH dataset aims at providing an upper bound on the average pairwise graph edit distance.  Our framework provides a lower bound of 21.6 on the average pairwise graph edit distance of PAH, which appears to be the best available bound to date.  In obtaining these results, we combine invariant convex sets based on the Schur-Horn orbitope, the Motzkin-Straus approximation of the inverse stability number, and the Goemans-Williamson approximation of the maximum-cut value.

\paragraph{\textbf{Notation}}
We denote the normal cone at a point $x \in \mathcal{C}$ of a closed convex set $\mathcal{C}$  by $\mathcal{N}_{\mathcal{C}}(x)$.  The projection map onto a subspace $\mathcal{E}\subset\mathbb{R}^n$ is denoted by $P_\mathcal{E}: \mathbb{R}^n \rightarrow \mathbb{R}^n$.  For a collection of subspaces $\mathcal{E}_i \subset \mathbb{R}^n, ~ i\in\{1,\dots,m\}$, the operator $\projop_{ij}: \mathbb{S}^{n} \rightarrow \mathbb{S}^{n}$ is defined as $\projop_{ij} := P_{\mathcal{E}_j} \otimes P_{\mathcal{E}_i}$, i.e., $\projop_{ij}(A) = P_{\mathcal{E}_i} A P_{\mathcal{E}_j}$. The restriction of a (usually self-adjoint) linear map $f: \R^n \to \R^n$ to an invariant subspace $\eig$ of $f$ is denoted by $f|_{\eig}: \eig \to \eig$.



\section{Theoretical Guarantees for the Schur-Horn Orbitope Constraint} \label{SectionSchurHorn}

In this section, we give theoretical guarantees that describe conditions under which employing the Schur-Horn orbitope as an invariant convex constraint set in \eqref{Optimization P} leads to the associated lower bound on the graph edit distance being \emph{tight}, i.e., the optimal value of \eqref{Optimization P} equals the graph edit distance.  Concretely, we consider conditions on a graph $\G$ and the structure of the edits that transform $\G$ to another graph $\G'$ so that $\ged(\G,\G') = \ged_{LB}(\G,\G'; \mathcal{C}_{\mathcal{SH}(\G)})$.  We begin with a description of our main theoretical results in Section~\ref{SectionMainResult}, some consequences of these results for specific graph families in Section~\ref{SHConseq}, and finally an experimental demonstration on the utility of the Schur-Horn orbitope on stylized problems in Section \ref{SectionSHExperiment}.  The proofs of the results of this section are deferred to Section \ref{SHProofs}.


As the normal cones at extreme points of the Schur-Horn orbitope play a prominent role in the optimality conditions of \eqref{Optimization P}, we state the relevant result here:
\begin{lemma} \label{LemmaNormalCone} \cite{candogan2018finding}
Let $\G$ be any unweighted graph with $m$ eigenvalues.  A matrix $W$ is an extreme point of $\mathcal{SH}(\G)$ if and only if it has the same eigenvalues as $\G$ (counting multiplicity). Further, the relative interior of the normal cone $\text{relint(}\mathcal{N}_{\mathcal{SH}(\G)}(W))$ at an extreme point $W$ consists of those matrices $Q$ that satisfy the following conditions:
\begin{enumerate}
\item $Q$ is simultaneously diagonalizable with $W$,
\item $\lambda_{\min}(Q|_{\mathcal{E}_i}) > \lambda_{\max}( Q|_{\mathcal{E}_{i+1}}) ~ \forall\,i \in\{1,\dots,m-1\}$,
\end{enumerate}
where $\mathcal{E}_i$ for $i\in\{1,\dots,m\}$ are eigenspaces of $W$ ordered such that the corresponding eigenvalues are sorted in a decreasing order.
\end{lemma}
From Lemma \ref{LemmaNormalCone} we observe that the relative interior of the normal cones of the Schur-Horn orbitope at extreme points are `larger' if the underlying graph $\G$ consists of few distinct eigenvalues.  This observation along with various properties of the eigenspaces of $\G$ play a prominent role in the analysis in this section.


\subsection{Main Results}\label{SectionMainResult}

We present here the statements of our main theoretical results concerning the performance of the Schur-Horn orbitope as a constraint set in \eqref{Optimization P}.  In addition to various structural properties of $\G$, our results are described in terms of a parameter $d$ that denotes the maximum number of deletions/additions of edges that are incident to any vertex of $\G$.  Informally, we should expect the Schur-Horn orbitope constraint to be effective in exactly computing the graph edit distance if a matrix representing the edits from $\G$ to $\G'$ has only a small amount of its energy on each of the eigenspaces of $\G$.  The reason behind this observation is that if the edits were largely concentrated in the eigenspaces of $\G$, then the eigenspaces of $\G'$ would be close to those of $\G$.  This would result in an identifiability problem from the perspective of the Schur-Horn orbitope, which is based purely on the spectral properties of $\G$.  To formalize this notion, we present the following definition which plays a key role in our analysis:
\begin{definition}
Let $\G$ be a graph on $n$ vertices with $m$ distinct eigenvalues.  Let $P_i, ~ i=1,\dots,m$ represent projection maps onto the eigenspaces of $\G$ indexed by decreasing order of the corresponding eigenvalues and let $\projop_{ii} = P_i \otimes P_i$.  Fix a positive integer $d$ and $\alpha\in[0,1]^m$.  Define the parameter $\xi(\alpha,d,\G)$ to be the smallest value
\begin{align*}
\norm{ [I - \sum_{i=1}^m \alpha_i \projop_{ii}] (W)  }_\infty \leq \xi(\alpha,d,\G) \norm{W}_\infty.
\end{align*}
for all $W \in \mathbb{S}^n$ with at most $d$ nonzero entries per row/column.
\end{definition}
\begin{remark} The maps $P_i$ represent projections onto eigenspaces of an adjacency matrix representing $\G$, but we simply refer to these as eigenspaces of $\G$ with an abuse of terminology.  The reason is that the quantity $\xi(\alpha,d,\G)$ is a graph parameter (for each fixed $\alpha, d$) and does not depend on a specific labeling of the vertices of $\G$.
\end{remark}
\begin{remark} The parameter $\xi(\alpha,d,\G)$ is a restricted version of the induced (entrywise) infinity norm $\norm{I - \sum_i^m \alpha_i \projop_{ii}}_{\infty\rightarrow \infty}$, with the key difference being that $\xi(\alpha,d,\G)$ computes the induced gain of the operator $I - \sum_i^m \alpha_i \projop_{ii}$ restricted to inputs that have at most $d$ nonzeros per row/column.
\end{remark}


The quantity $\xi(\alpha,d,\G)$ helps quantify the idea described previously about the energy of the edits not being confined excessively to the eigenspaces of $\G$.  As the specific edit pattern is not known in advance, this quantity is agnostic to the particular edits and is parametrized only in terms of the maximum number of edge deletions/additions that are incident to any vertex.  In our main results described next, larger values of $\xi$ make it harder to satisfy our sufficient conditions on tightness of our lower bounds.  As the value of $\xi$ depends on the selection of the parameter $\alpha$, our main results allow for flexibility in the choice of this parameter, and we describe in Section \ref{SHConseq} how specific choices lead to concrete consequences on the exactness of the relaxation \eqref{Optimization P} with the Schur-Horn orbitope constraint for various graph families.  We present next a result that establishes basic optimality conditions of the convex program \eqref{Optimization P}:

\begin{lemma} \label{LemmaOptimalityConditions}
Let $\G$ be a graph on $n$ vertices with $m$ distinct eigenvalues, and let $\G'$ be a graph that is obtained from $\G$ via edge deletions/additions such that each vertex is incident to at most $d$ edits.  Let $A, A+E^* \in \mathbb{S}^n$ represent the graphs $\G$ and $\G'$, respectively; that is, $E^*$ consists of at most $d$ nonzeros per row/column.  Let $\Omega \subset \mathbb{S}^n$ denote the subspace consisting of all matrices with nonzeros contained within the support of $E^*$.  Suppose a vector $\alpha\in [0,1]^m$ and a matrix $Q\in\mathbb{S}^{n}$ satisfy the following conditions:
\begin{enumerate}
	\item $\pomega( Q ) = \sign(E^*)$,
	\item $ || \pomegac(Q) ||_\infty < 1$, \label{LemmaOptimalityConditionPomegaC}
	\item $Q \in \text{relint} (\mathcal{N}_{\mathcal{SH}(\G)}(A) )$,
	\item $\xi(\alpha,d,\G)<1 $.
\end{enumerate}
Then we have that the convex relaxation \eqref{Optimization P} with the Schur-Horn orbitope constraint exactly computes the edit distance between $\G$ and $\G'$, i.e., $\ged(\G,\G') = \ged_{LB}(\G,\G'; \mathcal{C}_{\mathcal{SH}(\G)})$, with the optimal solution being unique and achieved at a matrix that specifies an optimal set of edits.
\end{lemma}
\begin{proof}
From the first three conditions, one can conclude that the pair $(A,E^*)$ is an optimal solution of \eqref{Optimization P} by a direct application of the KKT conditions.  Uniqueness can be established by standard arguments regarding transverse intersections of the subspace $\Omega$ and the invariant spaces of $\G$; see  \cite[Proposition 2]{chandrasekaran2012convex} and \cite[Lemma 6]{chen2013low}.
\end{proof}

Conditions 1, 2, and 3 of this lemma essentially require that the subdifferential at a matrix specifying the edits with respect to the $\ell_1$ norm has a nonempty intersection with the relative interior of the normal cone at an adjacency matrix representing $\G$ with respect to the Schur-Horn orbitope.  In papers on the topic of low-rank matrix completion and matrix decomposition \cite{candes2009exact, chandrasekaran2011rank}, a convenient approach to ensuring that such types of conic intersection conditions can be satisfied is based on requiring that nullspace (the eigenspace corresponding to an eigenvalue of zero) of the low-rank matrix is suitably `incoherent', i.e., that there are no elements of this nullspace with energy concentrated in a single location.  In our context, all of the eigenspaces of $\G$ play a role rather than just a single distinguished eigenspace, and accordingly we describe next a weighted form of an incoherence-type condition:
\begin{definition}
Let $\G$ be a graph on $n$ nodes with $m$ distinct eigenvalues and let $P_1,\dots,P_m  \in \mathbb{S}^n$ denote the projection matrices onto the associated eigenspaces indexed by decreasing order of the corresponding eigenvalues.  Fix any $\gamma \in \mathbb{R}^m$.  We define the parameter $\rho(\gamma,\G)$ as follows:
\begin{align*}
\rho(\gamma,\G) := \norm{\sum_{i=1}^m \gamma_i P_i}_\infty.
\end{align*}
\end{definition}
Here the matrix $\|\cdot\|_\infty$ norm is the largest entry of the argument in magnitude.  In the literature on inverse problems involving low-rank matrices, one typically considers the infinity norm of the projection map onto the nullspace as well as variants of this quantity.  Thus, in this sense the parameter $\rho(\gamma,\G)$ is a weighted generalization that is more suited to our setup.  We state next our main theorem in terms of sufficient conditions involving the two parameters we have introduced in this section:

\begin{theorem}\label{TheoremMainTheorem} Let $\G$ be a graph on $n$ vertices with $m$ distinct eigenvalues, and let $\G'$ be a graph that is obtained from $\G$ via edge deletions/additions such that each vertex is incident to at most $d$ edits.  Suppose the following two conditions are satisfied for some $\gamma\in \mathbb{R}^m$ and $\alpha\in [0,1]^m$:
\begin{enumerate}
\item $2 \, \xi(\alpha,d,\G) + \rho(\gamma,\G) < 1$, \label{TheoremConditionInfinityNorm}
\item $  \frac{(\alpha_{i}+\alpha_{i+1}) \,(1+ \rho(\gamma,\G))\,d}{1-\xi(\alpha,d,\G)} < \gamma_{i+1}-\gamma_i, ~ \forall\,i\,\in\,\{1,\dots,m-1\}$.\label{TheoremConditionEigenSep}
\end{enumerate}
Then the convex relaxation \eqref{Optimization P} with the Schur-Horn orbitope constraint exactly computes the edit distance between $\G$ and $\G'$, i.e., $\ged(\G,\G') = \ged_{LB}(\G,\G'; \mathcal{C}_{\mathcal{SH}(\G)})$, with the optimal solution being unique and achieved at a matrix that specifies an optimal set of edits.
\end{theorem}

%
%

This theorem states that the relaxation \eqref{Optimization P} with the Schur-Horn orbitope constraint set succeeds in calculating the graph edit distance exactly if 1) $d$ is small enough, 2) there exists a vector $\alpha$ with small entries such that $\xi(\alpha,d,\G)$ is also suitably small, and 3) there exists an ordered vector $\gamma$ with well-separated entries that yields a small value of $\rho(\gamma,\G)$.  As discussed in the next subsection, graphs with a small number of well-separated eigenvalues offer an ideal candidate.  Specifically, for such graph families, we give concrete consequences in terms of bounds on the maximum number $d$ of edits per vertex via particular choices for $\alpha$ and $\gamma$ in Theorem~\ref{TheoremMainTheorem}.


\subsection{Consequences for Graph Families with Few Eigenvalues} \label{SHConseq}
Theorem $\ref{TheoremMainTheorem}$ constitutes our most general result on the tightness of the Schur-Horn orbitope constraint in computing the graph edit distance when employed as a constraint set in the context of \eqref{Optimization P}.  The generality of the result stems from the wide range of flexibility provided by the vectors $\gamma$ and $\alpha$. In Corollary \ref{CorollaryMain}, we consider specific choices of these parameters to obtain concrete bounds in terms of graph parameters that can be computed easily:

\begin{corollary}$\label{CorollaryMain}$
Let $\G$ be a vertex-transitive graph on $n$ vertices consisting of $m$ distinct eigenvalues, and let $\kappa$ denote the multiplicity of the eigenvalue with the second-highest multiplicity.  Suppose $\G'$ is a graph on $n$ vertices that can be obtained from $\G$ with the addition or removal of at most $d$ edges incident to each vertex of $\G$.  Then there exists a constant $c$ depending only on $m$ so that the optimal value $\ged_{LB}(\G,\G'; \mathcal{C}_{\mathcal{SH}(\G)})$ of \eqref{Optimization P} equals $\ged(\G,\G')$ provided
\begin{align*}
d \leq c \frac{n}{\kappa}.
\end{align*}
\end{corollary}

The particular dependence on the multiplicity of the eigenvalue with second-largest multiplicity is due to the choices of $\alpha$ and $\gamma$ in Theorem \ref{TheoremMainTheorem} that we have employed in our proof; see Section \ref{SHProofs} for more details.  In the sequel we give consequences of this result for specific graph families in which the number of distinct eigenvalues is small (for example, three or four).  In the context of such graphs, the relaxation \eqref{Optimization P} is tight even when the number of edits per vertex is large so long as the value of $\kappa$ is suitably small.  Indeed, for several graph families we observe that Corollary \ref{CorollaryMain} produces `order-optimal' bounds as the largest value of $d$ that is allowed is on the same order as the degree of the underlying graphs.

\begin{figure}[hbt]
\centering
\subcaptionbox{\label{FigureHamming}}{\includegraphics[width=0.22\textwidth]{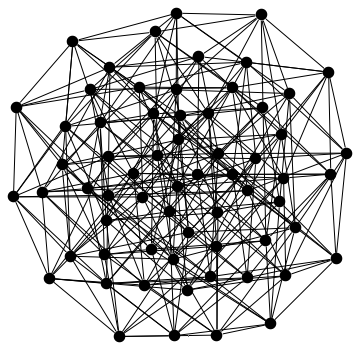}}  \hspace{0.35in}
\subcaptionbox{\label{FigureT9}}{\includegraphics[width=0.22\textwidth]{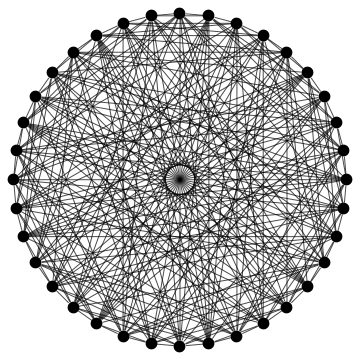}} \hspace{0.35in}%
\subcaptionbox{\label{FigureGQ24}}{\includegraphics[width=0.22\textwidth]{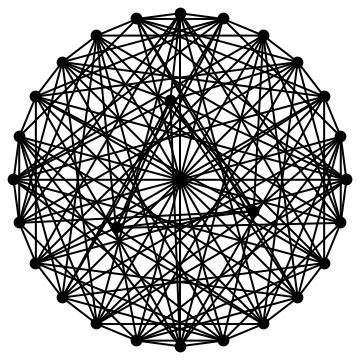}}
\caption{From left to right:  Hamming graph H(3,4), 9-Triangular graph, generalized quadrangle-(2,4) graph.}
\end{figure}

\paragraph{Johnson Graphs} A Johnson graph $J(k,\ell)$ with $\ell>0$ is a graph on $n=\binom{k}{\ell}$ vertices that correspond to the $\ell$-element subsets of a set of $k$ elements.  Two vertices of a Johnson graph are connected if the corresponding subsets of these vertices contain $\ell-1$ common elements. The Johnson graph $J(k,\ell)$ is vertex-transitive and contains $\ell+1$ distinct eigenvalues. For $k\geq 2\ell$ and $j\in\{0,\dots,\ell\}$, the multiplicity of its $j$'th eigenvalue is $\binom{k}{j}$ -$\binom{k}{j-1}$ for $j>0$ and one for $j=0$.  For small values of $\ell$, the multiplicity of the second most repeated eigenvalue is about $k^{\ell-1}$.
As a result, for small fixed values of $\ell$, Corollary \ref{CorollaryMain} states that the convex relaxation \eqref{Optimization P} is tight provided
\begin{align*}
d \lesssim n^{\frac{1}{\ell}}.
\end{align*}



\paragraph{Kneser Graphs} A Kneser graph $K(k,\ell)$ with $\ell>0$ shares certain aspects with Johnson graphs.  Specifically, the vertices of $K(k,\ell)$ coincide with the $\ell$-element subsets of a set of $k$ elements, as with Johnson graphs.  However, two vertices of a Kneser graph are connected if the subsets corresponding to these vertices are disjoint. Kneser graphs are vertex-transitive, and their eigenvalues exhibit the same multiplicities as those of the Johnson graphs $J(k,\ell)$. As a result, for small fixed values of $\ell$, Corollary \ref{CorollaryMain} implies that the relaxation \eqref{Optimization P} is tight provided:
\begin{align*}
d\lesssim n^{\frac{1}{\ell}}.
\end{align*}

\paragraph{Hamming Graphs} A Hamming graph $H(\ell,q)$ consists of $q^\ell$ vertices (see Figure \ref{FigureHamming} for a depiction of $H(3,4)$). Each vertex of $H(\ell,q)$ corresponds to a sequence of length $\ell$ from a set with $q$ distinct elements. Two vertices are connected if their associated sequences differ in exactly one coordinate, i.e., their Hamming distance is equal to 1.  Hamming graphs are vertex-transitive, and the spectrum of $H(\ell,q)$ consists of $\ell+1$ distinct eigenvalues with multiplicities $\binom{\ell}{i} (q-1)^i$, $i\in\{0,\dots,\ell\}$.  Therefore, for a small fixed value of $\ell$, Corollary \ref{CorollaryMain} states that the relaxation \eqref{Optimization P} is tight provided:
\begin{align*}
d \lesssim n^{\frac{1}{\ell}}.
\end{align*}


\paragraph{Vertex-Transitive Strongly Regular Graphs} A strongly regular graph on $n$ vertices with degree $r$ is defined by the property that every pair of adjacent vertices has $d_a$ common neighbors and every pair of nonadjacent vertices has $d_{na}$ common neighbors.  Such graphs are generally denoted $srg(n,r,d_a,d_{na})$.  Due to their rich algebraic structure, strongly regular graphs have only three distinct eigenvalues with multiplicities equal to one and $\frac{1}{2 } \Big[ (n-1) \pm \frac{2r + (n-1)(d_a-d_{na}) }{ \sqrt{ (d_a-d_{na})^2 + 4 (r-d_{na})   }  }  \Big]$. Furthermore, many strongly regular graphs are also vertex-transitive and as a result, our Corollary \ref{CorollaryMain} is applicable.  We highlight two prominent examples:
\begin{itemize}
\item A $k$-Triangular graph $T_k$ on $n= \binom{k}{2}$ vertices is a vertex-transitive strongly regular graph with parameters $srg(k(k-1)/2,2(k-2),k-2,4)$ (in fact, $T_k$ is also isomorphic to the Johnson graph $J(k,2)$); see Figure \ref{FigureT9} for the $9$-Triangular graph.  Corollary \ref{CorollaryMain} states that the convex relaxation (\ref{Optimization P}) is tight provided:
\begin{align*}
d \lesssim n^{\frac{1}{2}}.
\end{align*}
Incidentally, the degree of $T_k$ also scales as $n^{\frac{1}{2}}$; as a result, Corollary \ref{CorollaryMain} is tight for this family up to constant factors.

\item A generalized quadrangle is an incidence relation satisfying certain geometric axioms on points and lines. A generalized quadrangle of order $(s,t)$ gives rise to a strongly regular graph with parameters $srg((s+1)(st+1), s(t+1), s-1, t+1)$ denoted by $GQ(s,t)$ on $n = (s+1)(st+1)$ vertices -- see Figure \ref{FigureGQ24} for an illustration of the vertex-transitive graph $GQ(2,4)$.  Considering generalized quadrangle graphs $GQ(s,s^2)$ when they are vertex-transitive, Corollary \ref{CorollaryMain} implies that the relaxation \eqref{Optimization P} is tight provided
\begin{align*}
d \lesssim n^{\frac{1}{4}}.
\end{align*}
\end{itemize}
In Section \ref{SectionSHExperiment} we demonstrate the utility of our framework via numerical experiments on edit distance problems involving the graphs $T_9$ and $GQ(2,4)$.


\begin{figure}[hbt]
\centering
\subcaptionbox{\label{FigureT9SuccessProb}}{\includegraphics[width=0.45\textwidth]{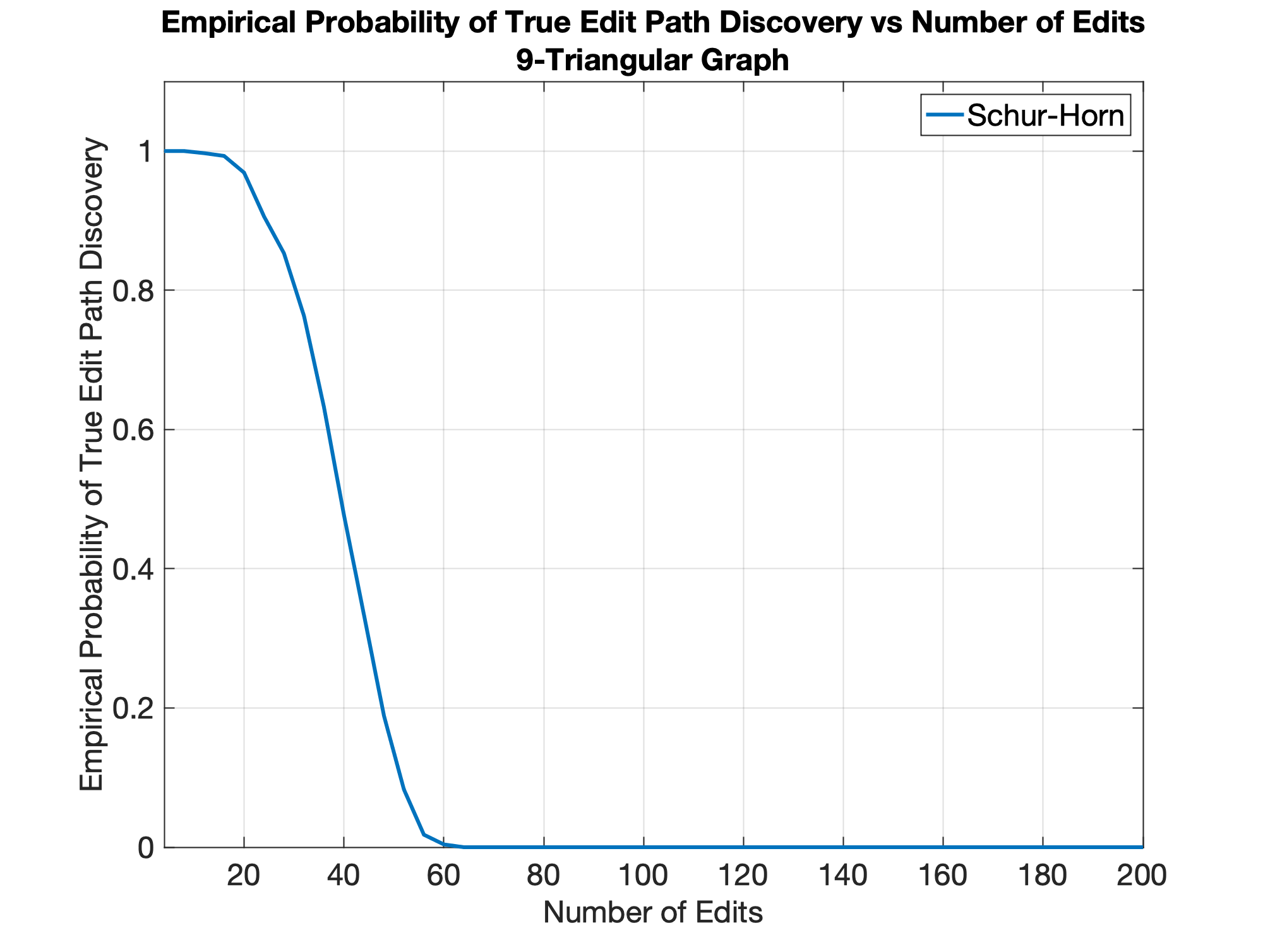}} \hspace{0.1in}%
\subcaptionbox{\label{FigureT9AverageRatio}}{\includegraphics[width=0.45\textwidth]{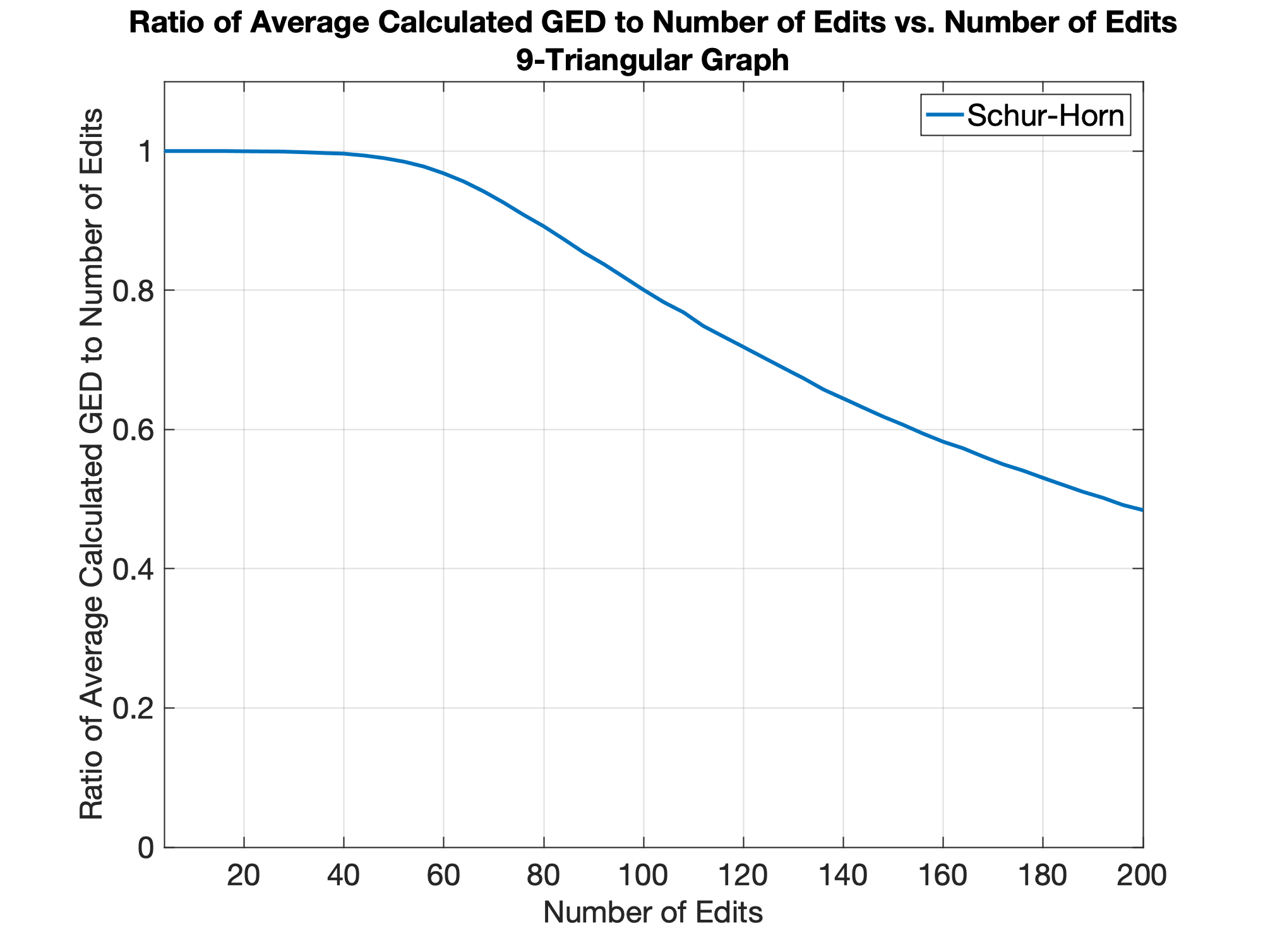}}
\subcaptionbox{\label{FigureGQ24SuccessProb}}{\includegraphics[width=0.45\textwidth]{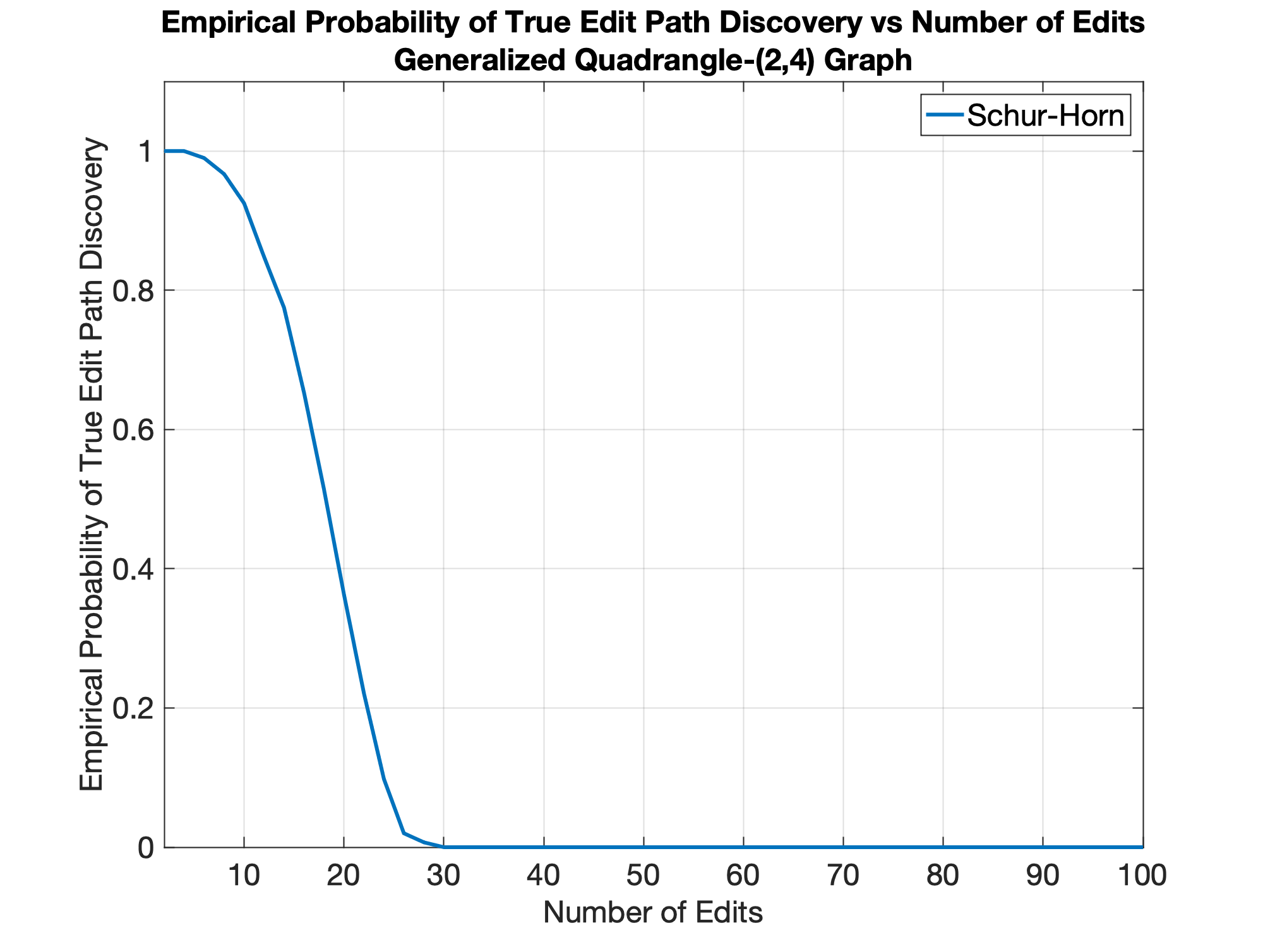}}
\subcaptionbox{\label{FigureGQ24AverageRatio}}{\includegraphics[width=0.45\textwidth]{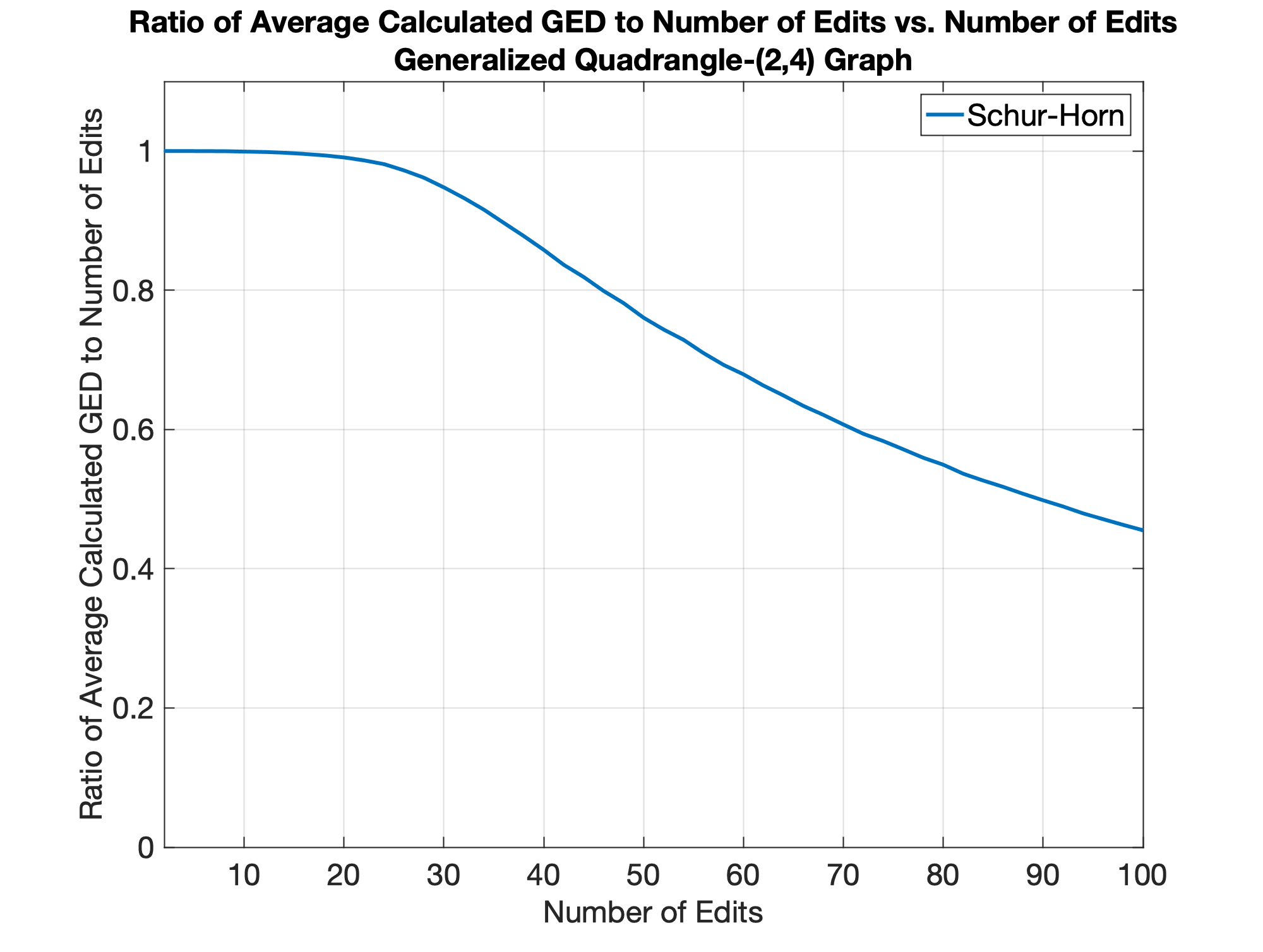}}
\caption{Performance of our framework (\ref{Optimization P}) with the Schur-Horn constraint. Left, empirical probability of discovering the true graph edit path. Right, ratio of average calculated graph edit distance to number of edit operations. First row corresponds to the $9$-triangular graph and the second row to the generalized quadrangle-(2,4) graph.}\label{FigureSH}
\end{figure}

\subsection{Numerical Experiments} \label{SectionSHExperiment}

We demonstrate the utility of the Schur-Horn orbitope as a constraint set in \eqref{Optimization P} in obtaining bounds on the graph edit distance between graphs $\G$ and $\G'$.  In our experiments, we fix $\G$ to be either the 9-triangular graph $T_9$ (Figure \ref{FigureT9}) or the generalized quadrangle-(2,4) graph $GQ(2,4)$ (Figure \ref{FigureGQ24}) introduced previously.  The graph $T_9$ consists of 36 vertices and 252 edges and the graph GQ(2,4) consists of 27 vertices and 135 edges.  Both of these are strongly regular graphs.  In each case, the corresponding graph $\G'$ is obtained by adding/deleting edges randomly (both addition and deletion occur with equal probability) to achieve a desired number of edits.  When $\G$ is $T_9$ we vary the number of edits from four to 200 in increments of four, and when $\G$ is $GQ(2,4)$ we vary the number of edits from two to 100 in increments of two.  For each number of edits, we consider 1000 random trials and we report the probability that $\ged(\G,\G') = \ged_{LB}(\G,\G'; \mathcal{C}_{\mathcal{SH}(\G)})$ and the ratio of the average computed lower bound $\ged_{LB}(\G,\G'; \mathcal{C}_{\mathcal{SH}(\G)})$ to the number of edits.  In particular, we declare that $\ged(\G,\G') = \ged_{LB}(\G,\G'; \mathcal{C}_{\mathcal{SH}(\G)})$ if the infinity norm (maximum entrywise magnitude) of the difference between the optimal solution $\hat{E}$ and the true edit matrix $E^*$ is less than $0.01$.  The results are shown in Figure \ref{FigureSH} and they were obtained using the CVX parser \cite{grant2008cvx,grant2008graph} and the SDPT3 solver \cite{toh1999sdpt3}.  As these plots demonstrate, the convex relaxation \eqref{Optimization P} with a Schur-Horn orbitope as an invariant convex constraint set is tight when the number of edits is small and leads to effective lower bounds when the number of edits is large.

\section{Proofs of Results from Section \ref{SectionSchurHorn}} \label{SHProofs}

\subsection{Constructing a Dual Certificate}
We describe here a method for constructing a suitable dual certificate satisfying the conditions of Lemma~\ref{LemmaOptimalityConditions}, and we prove that this construction is valid whenever certain conditions involving the parameters $\xi$ and $\rho$ from Section \ref{SectionSchurHorn} are satisfied.  Our proofs are presented in the context of two intermediary lemmas, which are then used to prove Theorem \ref{TheoremMainTheorem}.  Specifically, our approach to constructing $Q \in \mathbb{S}^n$ that satisfies the requirements of Lemma~\ref{LemmaOptimalityConditions} is to express $Q$ as follows:
\begin{align*}
Q = R + \Delta.
\end{align*}
Here $R \in \mathbb{S}^n$ plays the role of a `reference' matrix that depends purely on the underlying graph $\G$, while $\Delta \in \mathbb{S}^n$ is a perturbation that additionally depends on the specific edits that transform $\G$ to $\G'$.  We begin by stating an easily-proved result that serves as the basis for our subsequent development:

\begin{lemma} \label{LemmaOptimalityConditionsSufficient}
Let $\G$ be a graph on $n$ vertices with $m$ distinct eigenvalues, and let $\G'$ be a graph that is obtained from $\G$ via edge deletions/additions such that each vertex is incident to at most $d$ edits.  Let $A, A+E^* \in \mathbb{S}^n$ represent the graphs $\G$ and $\G'$, respectively; that is, $E^*$ consists of at most $d$ nonzeros per row/column.  Let $\Omega \subset \mathbb{S}^n$ denote the subspace consisting of all matrices with nonzeros contained within the support of $E^*$.  Let $P_i \in \mathbb{S}^n, ~ i=1,\dots,m$ denote projection maps onto the eigenspaces of $A$ indexed by decreasing order of the corresponding eigenvalues.  Suppose a vector $\alpha\in [0,1]^m$, a vector $\gamma \in \mathbb{R}^m$, and a matrix $\Delta \in \mathbb{S}^n$ satisfy the following conditions with $R = \sum_i \gamma_i P_i$:
\begin{enumerate}	
	\item $\pomega( \Delta ) + \pomega(R) = \sign(E^*)$, \label{LemmaOptimalityConditionsSufficientEnumPomega}
	\item $||\pomegac(\Delta)||_\infty + ||\pomegac(R)||_\infty < 1$	\label{LemmaOptimalityConditionsSufficientEnumNorm},
	\item $\projop_{ij} (\Delta) = 0 $, $\forall i,j\in\{1,\dots,m\},\, i\neq j$, \label{LemmaOptimalityConditionsSufficientPij}	
	\item $||\projop_{ii}(\Delta)||_2 + ||\projop_{i+1,i+1} (\Delta)||_2 < \gamma_{i+1} - \gamma_i$, $\forall i \in \{1,\dots,m-1\},$
	\label{LemmaOptimalityConditionsSufficientEnumEgap}
	\item $\xi(\alpha,d,\G) < 1. \label{LemmaOptimalityConditionsSufficientXiBound} $
\end{enumerate}
Then the convex relaxation \eqref{Optimization P} with the Schur-Horn orbitope constraint exactly computes the edit distance between $\G$ and $\G'$, i.e., $\ged(\G,\G') = \ged_{LB}(\G,\G'; \mathcal{C}_{\mathcal{SH}(\G)})$, with the optimal solution being unique and achieved at a matrix that specifies an optimal set of edits.
\end{lemma}
\begin{proof}
	One can check that the conditions of Lemma \ref{LemmaOptimalityConditions} are satisfied by setting $Q = R + \Delta$.
\end{proof}						

%

This lemma highlights the role of the parameter $\gamma$, in particular demonstrating that larger separation among the values of $\gamma$ makes condition \ref{LemmaOptimalityConditionsSufficientEnumEgap} easier to satisfy but may also increase the value of $||\pomegac(R)||_\infty$, thus making condition \ref{LemmaOptimalityConditionsSufficientEnumNorm} potentially harder to satisfy.  

We now move on to the perturbation term $\Delta$.  As this matrix must satisfy several of the constraints discussed in Lemma~\ref{LemmaOptimalityConditionsSufficient}, its construction is somewhat delicate.  We build on the ideas developed in \cite{chen2013low} in the context of low-rank matrix recovery, but with certain adaptations that are crucial to our setting.  We construct $\Delta$ as an element in the range of an operator $\operator:\mathbb{S}^{n} \rightarrow \mathbb{S}^{n}$ that is parametrized by $\alpha \in [0,1]^m$:
\begin{align*}
\operator := \left(\sum_{i=1}^{m} \alpha_{i}\projop_{ii}\right)\pomega\left[ I - \left(I - \sum_{i=1}^m \alpha_i \projop_{ii}  \right)\pomega \right]^{-1}.
\end{align*}
All of the operators here are as defined before.  A point of the departure in the description of this operator relative to the ideas in \cite{chen2013low} is that our version allows for `fractional' contractions (as well as integral ones) based on the choice of $\alpha$.  When it is well-defined (i.e., the term involving the inverse is indeed invertible), the operator $\operator$ possesses a number of properties that lead to a convenient approach for constructing a suitable dual variable:
\begin{itemize}
\item[(P1)] $\projop_{ij} \operator  = 0$ $\forall i,j \in\{1,\dots,m\}$, $i\neq j$
\item[(P2)] $\pomega \operator  = \pomega$.
\end{itemize}
In the context of Lemma~\ref{LemmaOptimalityConditionsSufficient}, property (P1) ensures that that $\Delta$ is completely contained in a desired subspace, as stipulated by condition \ref{LemmaOptimalityConditionsSufficientPij} of Lemma \ref{LemmaOptimalityConditionsSufficient}.  Further, property (P2) implies that condition \ref{LemmaOptimalityConditionsSufficientEnumPomega} of Lemma \ref{LemmaOptimalityConditionsSufficient} is satisfied -- in particular, we make use of this property to ensure that $\Delta$ takes on a desired value when restricted to $\Omega$. Conditions \ref{LemmaOptimalityConditionsSufficientEnumNorm} and \ref{LemmaOptimalityConditionsSufficientEnumEgap} of Lemma \ref{LemmaOptimalityConditionsSufficient} require that the quantities $\norm{\projop_{ii}(\Delta)}_2$ and $\norm{\pomegac(\Delta)}_\infty$ to be sufficiently small -- these conditions are satisfied by the operator $\operator$ as well, as documented next:

\begin{lemma} \label{LemmaNormBounds}
Consider the same setup as in Lemma \ref{LemmaOptimalityConditionsSufficient}.  Fix any $\alpha\in[0,1]^m$ such that $\xi(\alpha,d,\G) < 1$.  Then the operator $\operator :\mathbb{S}^n \rightarrow \mathbb{S}^n$ is well-defined (i.e., the term containing the inverse is indeed invertible) and the following inequalities hold:
\begin{enumerate}
\item  $\norm{[\pomegac \operator] (X)}_\infty \leq \frac{\xi(\alpha,d,\G) \norm{X}_\infty}{1-\xi(\alpha,d,\G)}$, \label{LemmaNormBoundsInfinity}
\item $\norm{[ \projop_{ii} \operator ] (X)}_2 \leq \frac{\alpha_i d \norm{X}_\infty }{1-\xi(\alpha,d,\G)}$.\label{LemmaNormBoundsTwo}	
\end{enumerate}
\end{lemma}
In addition to providing upper bounds that serve as a foundation for the proof of our main theorem, Lemma~\ref{LemmaNormBounds} conveys the significance of the parameter $\xi(\alpha,d,\G)$.  Specifically, a suitably small value of $\xi(\alpha,d,\G)$ guarantees that $\operator$ is well-defined, along with the conclusion that elements in the range of $\operator$ have small infinity norm (when restricted to $\Omega$) and small operator norm (restricted to eigenspaces of $\G$).  The lemma also suggests that the operator norm of the restriction of $\operator(X)$ to any eigenspace of $\G$ scales with the corresponding entry of $\alpha$. Consequently, one can adjust $\alpha$ and $\gamma$ to ensure that every inequality in Lemma \ref{LemmaOptimalityConditionsSufficient} condition \ref{LemmaOptimalityConditionsSufficientEnumEgap} is satisfied.  Identifying the best values of $\alpha$ and $\gamma$ to achieve this may be accomplished in special cases based on underlying structure in $\G$, as demonstrated by Corollary~\ref{CorollaryMain} and the many concrete consequences that are described in Section \ref{SHConseq}.  In particular, in the proof of Corollary \ref{CorollaryMain}, we choose $\gamma$ such that the separation between consecutive $\gamma_i$'s is proportional to the sum of consecutive $\alpha_i$'s and we demonstrate that this approach yields easily-computable bounds based on properties of the underlying graph $\G$ on the highest number $d$ of tolerable edits per vertex.






\subsection{Proofs} \label{Proofs}
\subsubsection{Proof of Lemma \ref{LemmaNormBounds}}
Our proof is analogous to that of \cite[Lemma 8]{chen2013low}.  In order to avoid notational clutter, we denote  $( I - \sum_{i=1}^m \alpha_i \projop_{ii})$ as $\projop_T$.  Then from the definition of $\xi(\alpha,d,\G)$, we have:
\begin{align}
 \norm{[\projop_T \pomega] (X)}_\infty \leq \xi(\alpha,d,\G) \norm{X}_\infty.\label{eqnConvergence}
\end{align}
Due to the assumption that $\xi(\alpha,d,\G)<1$, we have that the series $I + \projop_{T}\pomega + \projop_{T}\pomega\projop_{T}\pomega + \dots$ converges geometrically with rate $\frac{1}{1-\xi(\alpha,d,\G)}$ and equals $(I - \projop_T \pomega)^{-1}$.  As such, the operator $\operator$ is well-defined.


Next, we proceed to the upper bounds.  First we have that:
\begin{align}
\norm{[\pomegac \operator] (X)}_\infty =& \norm{ [ \pomegac \projop_{T} \pomega (I-\projop_{T}\pomega)^{-1}] (X)  }_\infty \notag \\
\leq& \norm{ [\projop_{T} \pomega(I-\projop_{T}\pomega)^{-1}] (X) }_\infty \notag \\
\leq& \xi(\alpha,d,\G)\norm{ (I-\projop_{T}\pomega)^{-1} (X) }_\infty \notag \\
\leq& \frac{ \xi(\alpha,d,\G)\norm{ X    }_\infty }{1-\xi(\alpha,d,\G)} \notag .
\end{align}
One can check that the first equality holds based on a term-by-term comparison.  The first inequality follows by dropping the projection $\pomegac$.  Bounding the resulting quantity from above using $\xi(\alpha,d,\G)$ yields the second inequality.  The last inequality follows from the geometric convergence of $(I-\projop_{T}\pomega)^{-1}$.

Next we bound the quantity involving the operator norm:
\begin{align*}
\norm{[\projop_{ii} \operator] (X)}_2 & = \norm{[\alpha_i \projop_{ii}  \pomega ( I -  \projop_T\pomega )^{-1}] (X) }_2\\
&\leq \alpha_i \norm{  [\pomega ( I -  \projop_T\pomega )^{-1}] (X) }_2\\
&\leq \alpha_i d \norm{[\pomega ( I -  \projop_T\pomega )^{-1}] (X)}_\infty\\
&\leq \alpha_i d \norm{ [( I -  \projop_T\pomega )^{-1}] (X)}_\infty\\
&\leq  \frac{\alpha_i d \norm{X}_\infty }{1-\xi(\alpha,d,\G)}.
\end{align*}
The first inequality holds by dropping the projection $\projop_{ii}$.  The second inequality holds due to the fact that the operator norm of a matrix with at most $d$ entries per row/column can be bounded above by $d$ times the maximum element in magnitude of the matrix.  The third inequality holds by dropping the operator $\pomega$.  The final inequality follows from geometric convergence, as before.

\subsubsection{Proof of Theorem~\ref{TheoremMainTheorem}}
We prove that under the assumptions of this theorem the sufficient conditions of Lemma \ref{LemmaOptimalityConditionsSufficient} are satisfied.  Set $R = \sum_{i=1}^{m} \gamma_i P_i $ where $P_i \in\mathbb{S}^n$ is the projection matrix corresponding to the $i$'th eigenspace of $\G$.  Denote the edits by a matrix $E^* \in \mathbb{S}^n$, and let the subspace of matrices with nonzero entries contained inside the support of $E^*$ be denoted $\Omega$.  Set $M = \sign(E^*)-\pomega(R)$ and note that $M \in \Omega$.  Condition \ref{LemmaOptimalityConditionsSufficientXiBound} of Lemma \ref{LemmaOptimalityConditionsSufficient} is satisfied based on assumption \ref{TheoremConditionInfinityNorm} of Theorem~\ref{TheoremMainTheorem}.  As a result, the operator $\operator$ is well-defined by Lemma~\ref{LemmaNormBounds}.  Set $\Delta = \operator(M)$. We prove that $Q= R + \Delta$ satisfies the requirements of Lemma \ref{LemmaOptimalityConditionsSufficient}.

Condition \ref{LemmaOptimalityConditionsSufficientEnumPomega} of Lemma \ref{LemmaOptimalityConditionsSufficient}: One can check that:
\begin{align*}
\pomega(\Delta) + \pomega(R) = \pomega(\operator(M)) + \pomega(R) = \pomega(M)+\pomega(R) = \sign(E^*).
\end{align*}
Here the second equality holds due to property (P1) of the operator $\operator$.

Condition \ref{LemmaOptimalityConditionsSufficientEnumNorm} of Lemma \ref{LemmaOptimalityConditionsSufficient}: We have that:
\begin{align*}
\norm{\pomegac (\Delta) }_\infty + \norm{\pomegac(R)}_\infty \leq \frac{\xi(\alpha,d,\G) \norm{M}_\infty}{1-\xi(\alpha,d,\G)} + \norm{R}_\infty \leq  \frac{\xi(\alpha,d,\G)  + \rho(\gamma,\G)}{1-\xi(\alpha,d,\G)} < 1.
\end{align*}
The first inequality employed assertion \ref{LemmaNormBoundsInfinity} of Lemma \ref{LemmaNormBounds}, the second inequality follows from the triangle inequality and the definition of $M$, and the last inequality holds by assumption \ref{TheoremConditionInfinityNorm} of the theorem.

Condition \ref{LemmaOptimalityConditionsSufficientPij} of Lemma \ref{LemmaOptimalityConditionsSufficient}:  Follows from property (P2) of operator $\operator$.

Condition \ref{LemmaOptimalityConditionsSufficientEnumEgap} of Lemma \ref{LemmaOptimalityConditionsSufficient}:  One can check that:
\begin{align*}
\norm{\projop_{ii}(\Delta)}_2 + \norm{\projop_{i+1,i+1} (\Delta)}_2 &= \norm{ [\projop_{ii} \operator](M)}_2 + \norm{[\projop_{i+1,i+1}\operator](M)}_2 \\ &\leq \frac{(\alpha_i + \alpha_{i+1})  (1+\rho(\gamma,\G))\,d }{1-\xi(\alpha,d,\G)} <\gamma_{i+1} - \gamma_i, ~ \forall\,i\,\in\,\{1,\dots,m-1\}.
\end{align*}
Here the first inequality follows from assertion \ref{LemmaNormBoundsTwo} of Lemma \ref{LemmaNormBounds} and the triangle inequality, and the second inequality follows from the assumption of the theorem.


%
%
%
%
%
%
%

\subsubsection{Proof of Corollary \ref{CorollaryMain}}

For this proof we require the notion of incoherence of a subspace, which measures how well the subspace is aligned with the standard basis vectors.  This notion appears prominently in results on sparse signal recovery via convex optimization \cite{donoho2003optimally}.
\begin{definition}
Let $S\subseteq \mathbb{R}^n$ be a subspace and let $P_S$ be the corresponding projection onto $S$. The \emph{incoherence} of $S$ is denoted $\mu(S)$ and is defined as
\begin{align*}
\mu(S) := \max_{i} \| P_S e_i \|_2.
\end{align*}
Here $e_i$ is the $i$'th standard basis vector.
\end{definition}
For any projection matrix $P_S$, one can check that the inequality $\norm{P_S}_\infty \leq \mu(S)^2$ is satisfied.
\begin{remark} \label{RemarkIncoherence}
For vertex-transitive graphs, the diagonal entries of a projection matrix associated to any eigenspace of the graph are identical.  As a result, the incoherence of an eigenspace $\mathcal{E}$ of a vertex-transitive graph on $n$ vertices is equal to
\begin{align*}
\mu(\mathcal{E}) = \sqrt{\frac{\dim(\mathcal{E})}{n}}.
\end{align*}
\end{remark}

We now proceed to the proof of the corollary.  Denote the eigenspaces of $\G$ by $\mathcal{E}_i$ for $i\in\{1,\dots,m\}$ ordered by decreasing eigenvalue order.  Remark \ref{RemarkIncoherence} implies that $\mu(\mathcal{E}_i) = \sqrt{\frac{\dim(\mathcal{E}_i)}{n}}$. Denote the second largest coherence of the eigenspaces of $\G$ by $\bar{\mu} = \sqrt{\frac{\kappa}{n}}$, and denote the index of the eigenspace with the highest incoherence by $\ell$.  Set $\alpha_\ell =1$ and the remaining entries of $\alpha$ to 0. Furthermore, choose $\gamma$ such that:
\begin{align*}
\gamma_{i+1}-\gamma_{i} = c_1 \frac{\alpha_{i} + \alpha_{i+1}}{{\bar{\mu}}^2} + \epsilon, ~~\forall\,i\,\in\,\{1,\dots,m-1\},~\text{for some $c_1>0$, $\epsilon>0$.} 
\end{align*}
Here $c_1$ and $\epsilon$ are positive constants that can be as small as desired.  To establish condition \ref{TheoremConditionEigenSep} of Theorem \ref{TheoremMainTheorem}, we prove that the inequality below holds for all $i \in \{1,\dots,m-1\}$:
\begin{align*}
\frac{(\alpha_{i}+\alpha_{i+1}) \,(1+ \rho(\gamma,\G))\,d}{1-\xi(\alpha,d,\G)} \leq \gamma_{i+1}-\gamma_{i}- \epsilon =  c_1 \frac{\alpha_{i} + \alpha_{i+1}}{{\bar{\mu}}^2}, ~ \forall\,i\,\in\,\{1,\dots,m-1\}.
\end{align*}
Clearly, if $\alpha_{i}+\alpha_{i+1}=0$ for some $i$, then the corresponding inequality is satisfied. On the other hand, all the remaining inequalities can be collapsed to a single one by dividing both sides of all such inequalities  by $\alpha_{i}+\alpha_{i+1}$:
\begin{align*}
\frac{(1+\rho(\gamma,\G))d}{1-\xi(\alpha,d,\G)} \leq c_1\frac{n}{\kappa},
\end{align*}
a sufficient condition for which is:
\begin{align}
\frac{(1+\rho(\gamma,\G)) c }{1-\xi(\alpha,d,\G)} \leq c_1. \label{EqThmEigensepSimplified}
\end{align}
In the remainder of the proof, we show that our particular choice of $\gamma$ and $\alpha$ satisfy inequality \eqref{EqThmEigensepSimplified} and Theorem \ref{TheoremMainTheorem} condition \ref{TheoremConditionInfinityNorm}.

We bound $\rho(\gamma,\G)$ from above via a change of variable.  Setting $\tilde{\gamma} = \gamma_1 \frac{\bar{\mu}^2}{c_1}$, we have that:
\begin{align}
\rho(\gamma,\G) = \norm{\sum_{i=1}^{m} \gamma_i P_i}_\infty &=  \norm{\sum_{i=1}^{m}   \Big[ \gamma_1  + \sum_{j=1}^{i-1}( \frac{c_1(\alpha_j+\alpha_{j+1})}{\bar{\mu}^2}  + \epsilon )    \Big]   P_i}_\infty \label{rhoineq1}   \\
&\leq  \norm{\sum_{i=1}^m \Big[\gamma_1 + \sum_{j=1}^{i-1} \frac{c_1(\alpha_j+\alpha_{j+1})}{\bar{\mu}^2} \Big]P_i}_\infty   + \sum_{i=1}^m  \sum_{j=1}^{i-1} \epsilon \norm{P_i  }_\infty   \label{rhoineq2}  \\
&\leq \frac{c_1}{{\bar{\mu}}^2} \norm{\sum_{i=1}^m \Big[\tilde{\gamma}+ \sum_{j=1}^{i-1} (\alpha_j + \alpha_{j+1})\Big]P_i}_\infty +\epsilon c_3 \label{rhoineq3}\\
&\leq  \frac{c_1}{{\bar{\mu}}^2} \sum_{i=1}^m \Big[ | \tilde{\gamma}+ \sum_{j=1}^{i-1} (\alpha_j + \alpha_{j+1})|\norm{P_i}_\infty\Big] +\epsilon c_3 \label{rhoineq4}\\
&\leq c_1 c_2 +\epsilon c_3.  \label{EqnMuBar}
\end{align}
Here \eqref{rhoineq2} follows by grouping all terms with $\epsilon$ and using the triangle inequality, \eqref{rhoineq3} follows by the change of variables described above and bounding all the terms in the right summand from above by one, and \eqref{rhoineq4} follows from the triangle inequality.  We choose the remaining degree of freedom $\tilde{\gamma}$ to eliminate the contribution of the subspace with the highest incoherence parameter in the left summand. Consequently, (\ref{EqnMuBar}) follows by bounding the infinity norms of the remaining projection matrices from above by $\bar{\mu}^2$. Crucially, the fact that $\alpha \in [0,1]^m$ and $m$ are viewed as fixed enables us to bound the sum from above with positive constants $c_2$ and $c_3$ that depend only on $m$.


Next, we use our particular choice for $\alpha$ to bound $\xi(\alpha,d,\G)$ from above. In particular, for any $W\in\mathbb{S}^{n}$ we have:
\begin{align}
&\norm{  \big[(I-\sum_{i=1}^m \alpha_i \projop_{ii}) \pomega\big] (W)  }_\infty \notag \\
&= \norm{  \sum_{\substack{i=1\\i\neq\ell}}^m \Big(P_i \pomega (W) + \pomega (W) P_i - P_i \pomega (W) P_i\Big) - 	\sum_{\substack{i=1\\i\neq\ell}}^m\sum_{\substack{j=1\\j\neq i,\ell}}^m P_i \pomega(W) P_j			 }_\infty \notag \allowdisplaybreaks\\
&\leq \sum_{\substack{i=1\\i\neq\ell}}^m \Big(\norm{P_i \pomega (W)}_\infty +	\norm{\pomega (W) P_i }_\infty	+ \norm{P_i \pomega (W) P_i}_\infty\Big) + \sum_{\substack{i=1\\i\neq\ell}}^m\sum_{\substack{j=1\\j\neq i,\ell}}^m \norm{ P_i \pomega(W) P_j	}_\infty \notag\allowdisplaybreaks\\
&\leq \left[ \sum_{\substack{i=1\\i\neq\ell}}^m 2 \mu(\mathcal{E}_i)\sqrt{d} +  \left(\sum_{\substack{i=1\\i\neq\ell}}^m{\mu(\mathcal{E}_i)}^2  +  \sum_{\substack{i=1\\i\neq\ell}}^m\sum_{\substack{j=1\\j\neq i,\ell}}^m \mu(\mathcal{E}_i)\mu(\mathcal{E}_j)\right) d \right] \norm{W}_\infty \notag \allowdisplaybreaks \\
&\leq ( c_4 \bar{\mu} \sqrt{d} + c_5 \bar{\mu}^2 d ) \norm{W}_\infty \allowdisplaybreaks \notag\\
&= (c_4 \sqrt{\frac{\kappa d}{n}} + c_5 \frac{\kappa d}{n})\norm{W}_\infty\allowdisplaybreaks \notag \\
&\leq (c_4\sqrt{c} + c_5 c) \norm{W}_\infty, \label{EqnXiBound}
\end{align}
for some positive real numbers $c_4$ and $c_5$ depending only on $m$. Here the first equality is obtained by rearranging the sum in terms of the projection matrices $P_i$, the first inequality is due to the triangle inequality, and the second inequality is a consequence of the following inequalities:
\begin{align*}
\norm{P_i\pomega(W)}_\infty &\leq \mu(\mathcal{E}_i) \sqrt{d} \norm{W}_\infty, \\
\norm{\pomega(W)P_i}_\infty &\leq \mu(\mathcal{E}_i) \sqrt{d} \norm{W}_\infty, \\
\norm{P_i\pomega(W) P_j}_\infty &\leq \mu(\mathcal{E}_i) \mu(\mathcal{E}_j) {d} \norm{W}_\infty;
\end{align*}
which hold for all $i,j \in \{1,\dots,m\}$.

Equations (\ref{EqnMuBar}) and (\ref{EqnXiBound}) assert that $\rho(\gamma,\G)$ and $\xi(\alpha,d,\G)$ can lowered as desired by reducing the constants $c_1,\epsilon$ and $c$. Consequently, one can check that both condition \ref{TheoremConditionInfinityNorm} of Theorem \ref{TheoremMainTheorem} and equation \eqref{EqThmEigensepSimplified} (which implies condition \ref{TheoremConditionEigenSep} of Theorem \ref{TheoremMainTheorem}) can be satisfied by first choosing a sufficiently small $c_1$ and $\epsilon$ (both depending on $m$) to bound $\rho(\gamma,\G)$ from above, and then suitably choosing a sufficiently small $c$ depending on $m$, $c_1$ and $\epsilon$.

\section{Numerical Illustrations with Invariants based on Stable Sets and Cuts} \label{SectionExploratory}


In this section we evaluate the utility of two invariant convex sets based on (tractable relaxations of) the inverse of the stability number and the maximum cut value, both of which are described in Section~\ref{SectionConvexInvariantSets}.  Our investigation is via numerical experiments rather than theoretical bounds as in Section~\ref{SectionSchurHorn}.  The primary reason for this choice is that we do not have a detailed understanding of the face structure of the invariant convex sets considered in this section; in contrast, we have a precise (and convenient for the purposes of analysis) characterization of the geometry of the Schur-Horn orbitope, which played a crucial role in the theoretical results of the previous section.  Nonetheless, we pursue a systematic approach in the present section by identifying classes of graphs that are `brittle' in the sense that deleting / adding a small number of edges results in large changes in their stability number / maximum cut value.  Such graph families present excellent examples for which invariant convex sets based on the inverse of stability number and the maximum cut value are particularly well-suited to obtaining useful bounds on the graph edit distance.  More broadly, our discussion in this section highlights the larger point that our framework \eqref{Optimization P} can be tailored to the particular structural properties of the underlying graphs to yield useful lower bounds on the edit distance.

	
\subsection{Constraining the Inverse of the Stability Number} \label{SectionStabilityNumber}

The function $f(A)$ described in Section \ref{SectionConvexInvariantSets} is an efficiently computable lower bound on the inverse of the stability number of a graph, and further it is a concave graph invariant.  Consequently, super-level sets of this function provide tractable invariant convex sets that may be employed in our framework \eqref{Optimization P}.  Given a graph $\G$, we denote the associated set by $\mathcal{C}_{\mathcal{IS}(\G)}$:
\begin{align} \label{EqnCIS}
\mathcal{C}_{\mathcal{IS}(\G)} := \{M \in\mathbb{S}^n ~|~   f(M)\geq f(A) \} = \{ M \in\mathbb{S}^n ~|~ \exists \mu\in\mathbb{S}^n,\, \mu \geq 0 ,~I+M-\mu- f(A) 11^T \succeq 0    \}.
\end{align}
Here $A$ is any adjacency matrix representing $\G$.  From this description, it is immediately clear that for any edit to $\G$ that corresponds to an increase in the value of the function $f$, the constraint $\mathcal{C}_{\mathcal{IS}(\G)}$ is inactive.  Adding edges to a graph can only reduce the stability number, and hence can potentially only increase the inverse of the stability number.  Although the function $f$ is only a lower bound on the inverse of the stability number, it satisfies a similar monotonocity property in that the value of $f$ is non-decreasing with the addition of edges to a graph.  The following lemma formalizes matters by describing the tangent cone at an adjacency matrix of a graph $\G$ with respect to the set $\mathcal{C}_{\mathcal{IS}(\G)}$:
\begin{lemma}\label{LemmaInverseStabilityTC}
For any graph $\G$ on $n$ vertices and associated adjacency matrix $A \in \mathbb{S}^n$, let $\alpha^* = f(A)$, i.e., the value corresponding to the Motzkin-Straus relaxation of the inverse of the stability number.  Then we have that:
\begin{align*}
\mathcal{T}_{\mathcal{C}_{\mathcal{IS}(\G)}}(A) = \{ T \in\mathbb{S}^n \,|\,  \exists \, \mu,\,\Lambda\,\in\,\mathbb{S}^n,\, \mu \geq 0 ,\, \Lambda \succeq 0,~ T + I + A - \alpha^* 11^T = \mu + \Lambda \}.
\end{align*}
\end{lemma}
\begin{proof}
The proof of this lemma follows from a direct application of convex duality.
\end{proof}

The description of the tangent cone in Lemma \ref{LemmaInverseStabilityTC} is based on the dual of the cone of doubly nonnegative matrices; see \cite{berman1994nonnegative} for more details on this connection.  In particular, this lemma implies that entrywise nonnegative matrices belong to the tangent cone at an adjacency matrix $A$ representing a graph $\G$ with respect to the set $\mathcal{C}_{\mathcal{IS}(\G)}$; consequently, edits to $\G$ consisting purely of addition of edges are feasible directions with respect to the set $\mathcal{C}_{\mathcal{IS}(\G)}$ and for such edits this set does not provide useful lower bounds on the edit distance.  Thus, we investigate the utility of the constraint $\mathcal{C}_{\mathcal{IS}(\G)}$ in settings in which the edits consist mainly of edge deletions.  Such problems arise in the context of \emph{graph completion} in which the objective is to add edges to a given graph so that the resulting graph satisfies some desired property.

Building on this discussion, the constraint set $\mathcal{C}_{\mathcal{IS}(\G)}$ is most likely to be useful for graphs $\G$ in settings in which the deletion of even a small number of edges of $\G$ results in an increase in the stability number.  Graphs that have a large number of stable sets with cardinality equal to the stability number offer a natural prospect for further exploration.  Fortunately, such graphs have been studied in extremal graph theory literature, from which we quote the following result \cite{jou2000number}:
\begin{theorem} \label{TheoremExtremal} \cite{jou2000number}
For $s, n \in \mathbb{N}$ with $n \geq 6$, let
\begin{align*}
h(n) = \begin{cases}
2\times3^{s-1} + 2^{s-1},~\text{if}~n=3s,\\
3^{s} + 2^{s-1},~\text{if}~n=3s+1,\\
4 \times 3^{s-1} + 3\times 2^{s-2},~\text{if}~n=3s+2.
\end{cases}
\end{align*}
Let $\G$ be any connected graph on $n$ vertices, and denote the cardinality of the set of all maximum independent sets of $\G$ by $\phi(\G)$.  Then $\phi(\G) \leq h(n)$ with equality if and only if $\G$ is isomorphic to one of the graphs shown in Figure \ref{FigureExtremalGraph}.
\end{theorem}

\begin{figure}[hbt]
\centering
\subcaptionbox{\label{Extremal1}}{\includegraphics[width=30mm,height=41mm]{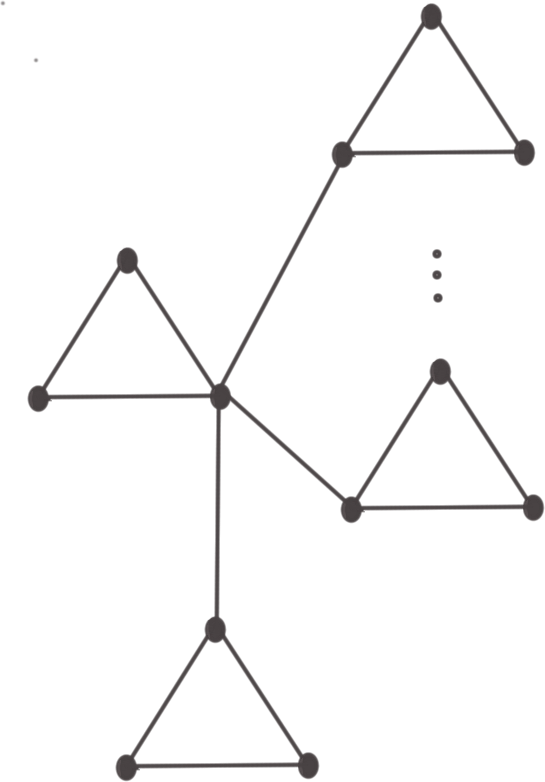}}\hspace{0.3in}%
\subcaptionbox{\label{Extremal2}}{\includegraphics[width=30mm,height=41mm]{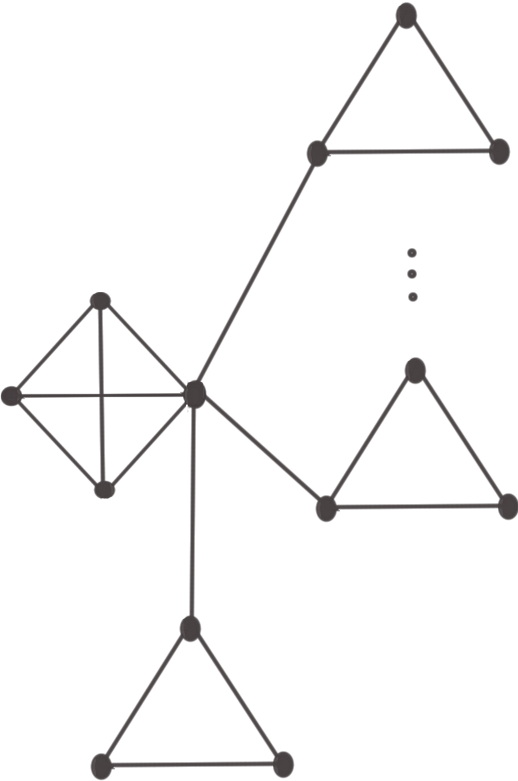}}\hspace{0.3in}%
\subcaptionbox{\label{Extremal3}}{\includegraphics[width=30mm,height=41mm]{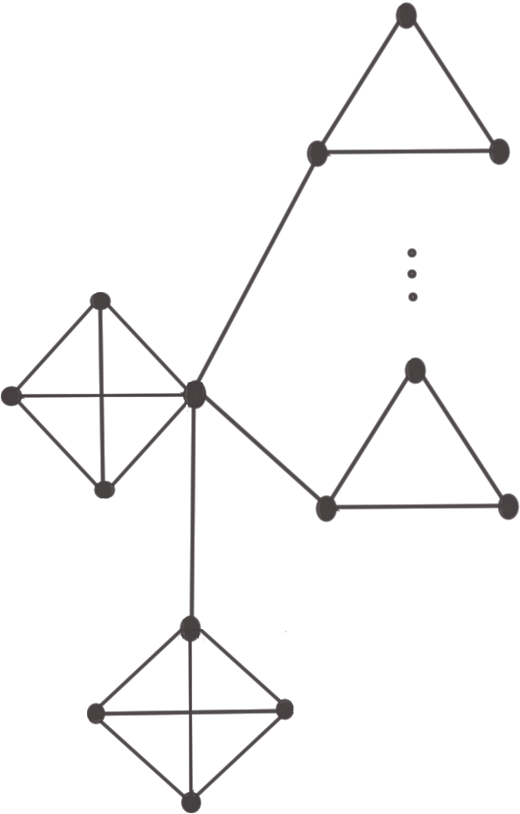}}
\caption{Left to right, $E(n)$ for $n = 3s$, $n = 3s+1$, $n = 3s+2$. For $n = 3s+r$, these graphs are formed by connecting $(s-r)\,K_3$'s and $r\, K_4$'s through edges connecting to a specific vertex.} \label{FigureExtremalGraph}
\end{figure}

This theorem states that the graphs $E(n)$ shown in Figure \ref{FigureExtremalGraph} are precisely the connected graphs that have the largest number of distinct maximum independent sets.  As such, they present a natural test case to investigate the utility of the constraint set $\mathcal{C}_{\mathcal{IS}(\G)}$ in providing bounds on the graph edit distance, at least in settings in which the edits are composed predominantly of edge deletions.  We illustrate here the results of numerical experiments conducted on the graph $E(30)$, which is a sparse graph with $39$ edges and $396$ nonedges.  The setup of this experiment is the same as that described in Section \ref{SectionSHExperiment} with one notable exception: in the present experiment, we assume asymmetric edits rather than symmetric edits so that 80\% of the edits are edge deletions while 20\% are edge additions.  We range the total number of edits from $5$ to $45$ with increments of $5$, and for each number of edits we repeat our experiment $1000$ times.  In each iteration, we obtain a bound on the graph edit distance between $E(30)$ and the modified graph using our framework (\ref{Optimization P}) with three different constraint sets: the Schur-Horn orbitope, the constraint set $\mathcal{C}_{\mathcal{IS}(\G)}$, and an invariant convex set based on the Goemans-Williamson relaxation of the maximum cut value (which is discussed in greater detail in the next subsection).  Figure \ref{FigureExtremalPhaseTran} reports the ratio of the average computed lower bound on the graph edit distance to the number of edit operations for each constraint set.  (The number of edits is an upper bound on the true graph edit distance.)  As one might expect, the relaxation based on the constraint $\mathcal{C}_{\mathcal{IS}(\G)}$ yields the best lower bounds of the three approaches.  Specifically, even when a majority of the edges of $E(30)$ are removed, the constraint set $\mathcal{C}_{\mathcal{IS}(\G)}$ continues to provide lower bounds that are at least 40\% of the number of edit operations. In contrast, the bounds provided by the Schur-Horn orbitope constraint are much weaker, and those obtained using the Goemans-Williamson relaxation of the maximum cut value are ineffective.


\begin{figure}[hbt]
\begin{center}
\includegraphics[scale=0.14]{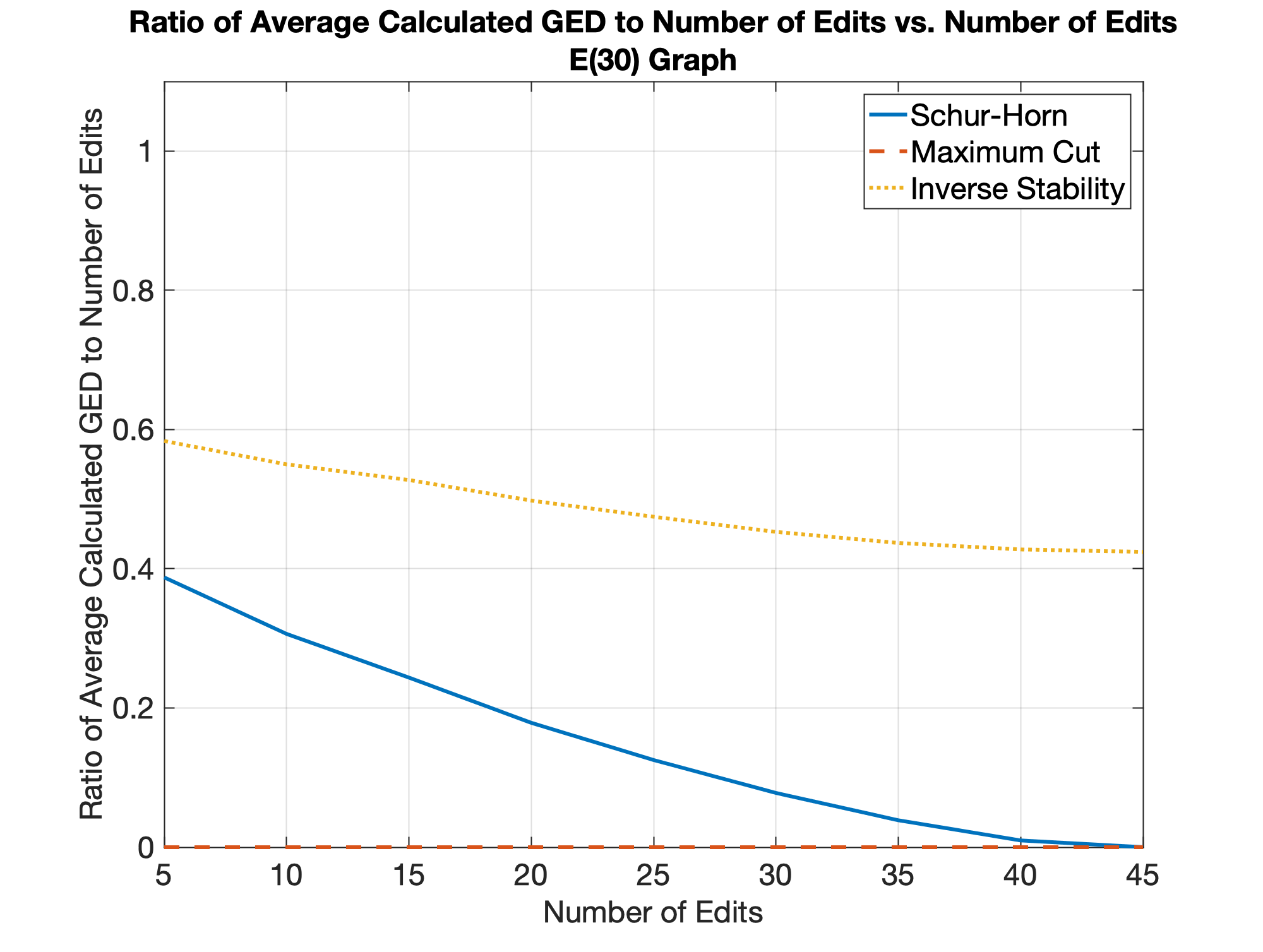}
\caption{Ratio of average computed lower bound on graph edit distance to number of edit operations. Experiment conducted on $E(30)$ graph.  The edit operations are 80\% edge deletions and 20\% edge additions.}\label{FigureExtremalPhaseTran}
\end{center}
\end{figure}



\subsection{Constraining the Maximum Cut Value} \label{SectionMaxcut}

In analogy with the inverse of the stability number, the function $g(A)$ due to Goemans and Williamson \cite{goemans1995improved} that is described in Section~\ref{SectionConvexInvariantSets} provides an efficiently computable upper bound on the maximum cut value of a graph.  As this function is invariant to conjugation of its argument by permutation matrices, its sublevel sets are invariant convex sets.  For a graph $\G$, we denote the associated set by $\mathcal{C}_{\mathcal{MC}(\G)}$:
\begin{align} \label{EqnCMC}
 \mathcal{C}_{\mathcal{MC}(\G)} :=& \{M\in\mathbb{S}^n ~|~ g(M)\leq g(A) \}  \notag \\
 =& \{ M\in\mathbb{S}^n~ | ~ \exists D\in\mathbb{S}^n \textrm{ diagonal,}~ M-D\succeq 0,\,  \frac{1}{4}\,\trace(M 11'-D) \leq g(A) \},
\end{align}
where $A$ is an adjacency matrix representing $\G$.  Reasoning in a similar manner as in the previous subsection, we observe that edits corresponding to a decrease in the value of the function $g$ represent feasible directions with respect to the set $\mathcal{C}_{\mathcal{MC}(\G)}$, and for such edits the constraint $\mathcal{C}_{\mathcal{MC}(\G)}$ is inactive.  Deleting edges from a graph reduces its maximum cut value, and one can check that directions represented by entrywise nonpositive matrices belong to the tangent cone at an adjacency matrix $A$ representing $\G$ with respect to $\mathcal{C}_{\mathcal{MC}(\G)}$.  Consequently, we should only expect the constraint $\mathcal{C}_{\mathcal{MC}(\G)}$ to potentially provide useful lower bounds on the graph edit distance in settings in which most of the edits to a graph $\G$ correspond to edge additions.  In some sense, this type of a graph inverse problem -- removing the smallest number of edges from a graph so that it satisfies a desired property -- is a complement of the graph completion problem discussed in the previous subsection.

\begin{figure}[t!]
    \centering
    \begin{subfigure}[t]{0.5\textwidth}
        \centering
        \includegraphics[scale=0.33]{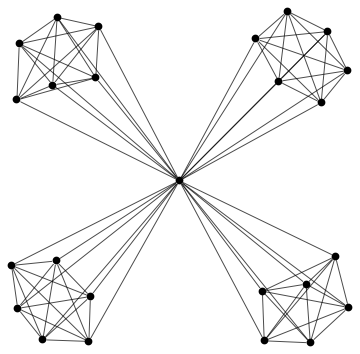}
		\caption{ Windmill graph $D(4,7)$}
	\end{subfigure}%
    ~
    \begin{subfigure}[t]{0.5\textwidth}
        \centering
        \includegraphics[scale=0.33]{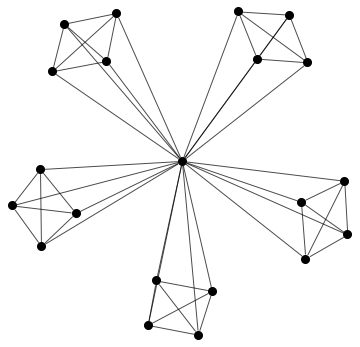}
        \caption{Windmill graph $D(5,5)$}
    \end{subfigure}
    \caption{Two sample Windmill graphs}\label{FigureWindmill}
\end{figure}

Building further on the preceding discussion, we remark that the constraint set $\mathcal{C}_{\mathcal{MC}(\G)}$ is most likely to be effective if adding even a small number of edges to $\G$ increases the value of the function $g$.  A prominent example of such graphs are the so-called Windmill graphs shown in Figure \ref{FigureWindmill}.  The Windmill graph $D(m,n)$ is constructed by taking $m$ copies of the complete graph $K_n$ and intersecting them at a single vertex.  Due to the ample amount of symmetry in these graphs, there are many partitions of the vertices into two sets that achieve the maximum cut value -- the number of  such partitions is $\binom{n-1}{n/2} ^m$ for even $n$ and $\binom{n}{(n-1)/2}^m$ for odd $n$.  Thus, Windmill graphs present a natural test family to evaluate the power of the constraint set $\mathcal{C}_{\mathcal{MC}(\G)}$ when the graph edits consist primarily of the addition of edges.  We present the results of numerical experiments on the Windmill graph $D(4,7)$ in a setting that closely mirrors the one in the previous subsection.  The Windmill graph $D(4,7)$ is a graph on 25 nodes with 84 edges and 216 non-edges.  The edits made to this graph consist mostly of edge additions -- 80\% are edge additions and the remaining 20\% are edge deletions.  We vary the number of edits from 10 to 200 with increments of 5 and consider 1000 random instances of perturbations for each number of edits.  For each problem instance, we obtain a lower bound on the edit distance by utilizing our framework (\ref{Optimization P}) with the Schur-Horn orbitope constraint, the Motzkin-Straus relaxation from the previous subsection, and the constraint $\mathcal{C}_{\mathcal{MC}(\G)}$.  We report the average ratio of the computed lower bound on the graph edit distance to the number of edit operations in Figure \ref{FigureMC}. (As before the number of edits is an upper bound on the graph edit distance.)

\begin{figure}[hbt]
\begin{center}
\includegraphics[scale=0.14]{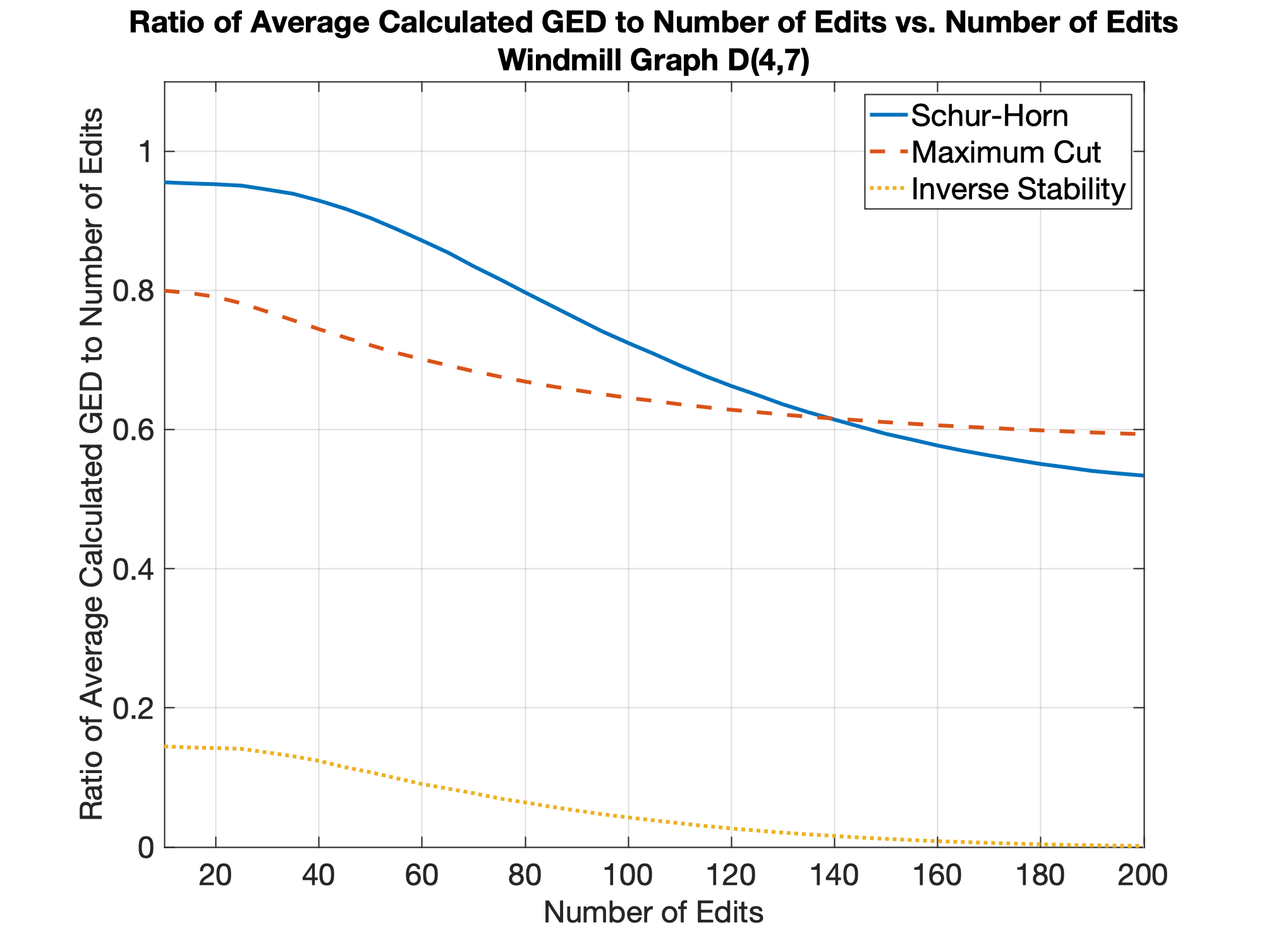}
\caption{Ratio of average computed lower bound on the graph edit distance to number of edit operations. Experiment conducted on Windmill graph $D(4,7)$. The edit operations are 80\% edge additions and 20\% edge deletions.}\label{FigureMC}
\end{center}
\end{figure}

From Figure \ref{FigureMC} we see that the Schur-Horn orbitope constraint produces the best lower bounds for the graph edit distance when the number of edits is small, whereas the constraint $\mathcal{C}_{\mathcal{MC}(\G)}$ produces the best lower bounds when the number of edits is large.  On average, both of these constraints provide bounds that are consistently better than 50\% of the total number of edits.  As the edits consist mainly of edge additions, the constraint based on the Motzkin-Straus relaxation of the inverse of the stability number performs poorly.

\section{Experiments with Real Data}\label{SectionRealExperiment}

In this section, we present experimental results that demonstrate the utility of our framework on real data.  We begin by introducing an extension of our framework to allow for edits that include vertex additions and deletions.  We then describe the bounds obtained on two widely studied datasets consisting of molecular structures.

\subsection{Enabling vertex additions and deletions}
\label{SectionExtendedFramework}


In many situations, one wishes to obtain bounds on the edit distance between two graphs consisting of different numbers of vertices.  In such cases one allows vertex insertions and deletions in addition to the usual operations of edge insertions and deletions that we've considered thus far.  To extend our framework to this setting, we allow an adjacency matrix to take on nonzero values on the diagonal to denote the presence or absence of a vertex.  Specifically, we consider a ``vertex-indexed adjacency matrix'' $A \in \mathbb{S}^n$ with entries equal to either zero or one and in which $A_{ij} = 1, ~ j \neq i$ implies that $A_{ii}=1$.  In words, a value of one on the $i$'th diagonal entry implies that a vertex corresponding to that index is `present' in the graph, and an edge being incident on a vertex implies that the vertex must be present in the graph.  (Note that a value of one on a diagonal entry does not represent a ``vertex weight'' but instead the presence of a vertex in a graph.) With this notation in hand, we are now in a position to describe a generalization of our framework that allows for vertex deletions and insertions.  Let $\G_1$ and $\G_2$ be two unweighted and unlabeled graphs on $n_1$ and $n_2$ vertices, respectively, and let $n := \max\{n_1,n_2\}$.  Letting $A_2 \in \mathbb{S}^n$ specify a vertex-indexed adjacency matrix for $\G_2$ with zeros on the diagonal corresponding to those indices that do not correspond to a vertex (when $n_2 < n$), consider the following convex optimization problem:
\begin{align}
\textrm{GED}_{LB}(\G_1, \G_2; \mathcal{C}_{\G_1}) &= \min_{X,E\in\mathbb{S}^n} \sum_{1\leq i \leq j \leq n} |E_{ij}| \notag \\
&~~~~~s.t. ~~~X+E = A_2	 \tag{$P_{ext}$} \label{ProgramExtended}\\
&~~~~~~~~~ ~~~X \in \mathcal{C}_{\G_1} \notag \\
&~~~~~~~~~ ~~~X_{ij} \leq X_{ii},~X_{ij} \leq X_{jj} ~\forall i,j \in \{1,\dots,n\}. \notag	
\end{align}
Here the set $\mathcal{C}_{\G_1}$ is an invariant convex set associated to $\G_1$, and the matrices $X$, $A_1$ and $A_2$ are to be interpreted as vertex-indexed adjacency matrices.  There are two main differences between the convex program \eqref{ProgramExtended} and the convex problem \eqref{Optimization P}.  The first is in the objective function in which we only sum the upper triangular elements of the matrix $E$ in the program \eqref{ProgramExtended}, as we do not wish to double-count the edge edits relative to vertex edits.  The second modification arises in the constraint in the last line of \eqref{ProgramExtended} based on the observation that if edges are incident to a vertex, then this vertex must be `present'.  Using a line of reasoning similar to that following the presentation of \eqref{Optimization P}, one can conclude that the optimal value of the convex program \eqref{ProgramExtended} provides a lower bound on the graph edit distance between $\G_1$ and $\G_2$ with the permissible edit operations being vertex additions/deletions and edge additions/deletions.  Finally, we note that our framework can also accommodate situations in which the cost of a vertex edit operation is different from that of an edge edit operation.



\subsection{Experimental Results on Chemistry Datasets} \label{SectionExperimentalResults}
We employ the convex program \eqref{ProgramExtended} to obtain lower bounds on graph edit distance problems arising in chemistry.  Specifically, we conduct experiments on two datasets known as the Polycyclic Aromatic Hydrocarbons (PAH) dataset and the Alkane dataset.\footnote{Available online at https://brunl01.users.greyc.fr/CHEMISTRY/}  Both of these datasets consist of unlabeled, unweighted graphs representing chemicals, with the vertices of the graphs corresponding to carbon atoms in a molecule and edges specifying bonds between two carbons.  These datasets have been used as benchmarks for evaluating the performances of graph edit distance algorithms; for example, see \cite{abu2017graph} and \cite{daller2018approximate} for comparisons of the performance of various algorithms on these datasets.  For each dataset, we compare upper bounds on the average edit distance taken over all pairs of graphs (obtained using other procedures) with lower bounds on the average obtained using our method.

The Alkane dataset consists of 150 unlabeled, acyclic graphs representing alkanes, with the number of vertices ranging from 1 to 10 vertices (the average is 8.9) and an average degree of 1.8.  As these graphs are relatively small in size, the average pairwise graph edit distance for this dataset can be calculated exactly using combinatorial algorithms such as the $A^*$ procedure \cite{hart1968formal}.  The PAH dataset consists of 94 graphs representing polycyclic aromatic hydrocarbons. As with the Alkane dataset, the vertices of the graphs in this dataset denote carbon atoms, and two vertices are connected if there exists a bond between the corresponding carbons.  Unlike the Alkane dataset, the chemicals in the PAH dataset represent large compounds: the smallest graph in PAH has 10 vertices, the largest graph has 28 vertices, and the average number of vertices is 20.7. The average degree of the graphs in the PAH dataset is 2.4.  Due to this larger size, calculating the exact average pairwise graph edit distance of the PAH dataset is prohibitively expensive.  In fact, to the best of our knowledge, the exact average pairwise graph edit distance of the PAH dataset is unknown to this date \cite{bougleux2017graph}. Consequently, obtaining guaranteed lower bounds on the average graph edit distance of the PAH dataset is especially useful as a way to compare to known average upper bounds.

For each pair of graphs, we employ the convex program \eqref{ProgramExtended} twice by switching the roles of $\G_1$ and $\G_2$, and take the larger optimal value as our lower bound.  In each case we utilize four different types of invariant convex set constraints: the Schur-Horn orbitope ($\mathcal{C}_\mathcal{SH}$), the Motzkin-Straus bound on the inverse of the stability number ($\mathcal{C}_\mathcal{IS}$), the Goemans-Williamson bound on the maximum cut value ($\mathcal{C}_\mathcal{MC}$), and finally the intersection of all these three constraints ($\mathcal{C}_\mathcal{SH}\cap \mathcal{C}_\mathcal{IS} \cap \mathcal{C}_\mathcal{MC} $).  In our experiments, we follow the convention adopted in the graph edit distance literature with these two datasets, namely that the cost of an edit operation is equal to three.\footnote{The reason for this choice in that community is that vertex/edge deletions/insertions are considered more significant edit operations than vertex/edge label substitutions which have a lower cost of one associated to them (in this paper, we do not consider such edits based on substitutions.)}  The average pairwise lower bounds obtained using our convex program \eqref{ProgramExtended} on the Alkane and PAH datasets are given in Table \ref{TableResults}.

\begin{table}[h] 
\begin{center}
\begin{tabular}{ |c|c|c|c|c|c| }
  \hline
  Dataset &
  Best known upper bound  & \multicolumn{4}{|c|}{ Lower bounds on the average pairwise GED via (\ref{ProgramExtended})}  \\
  \cline{3-6}
   Name &  on the average GED &  $\mathcal{C}_\mathcal{MC}$ & $\mathcal{C}_\mathcal{IS}$ & $\mathcal{C}_\mathcal{SH}$ & $\mathcal{C}_\mathcal{MC}\cap\mathcal{C}_\mathcal{IS} \cap \mathcal{C}_\mathcal{SH}$\\
  \hline
	Alkane & 15.3 (exact) \cite{daller2018approximate} & 4.66 & 6.12 & 9.58 & 10.72	\\
	\hline
	PAH & 29.8 \cite{daller2018approximate}  & 12.01 & 14.52 & 20.29 & 21.60	\\
	\hline
\end{tabular}
\end{center}
\caption{Average pairwise graph edit distances of the Alkane and PAH datasets. Edit operations are limited to edge and vertex additions and removals. Every edit operation incurs a cost of 3.}\label{TableResults}
\end{table}


There are a number of interesting aspects to these results.  For both datasets, a constraint based only on the Goemans-Williamson relaxation seems to produce the worst lower bounds (4.66 and 12.01), while the Schur-Horn orbitope constraint produces the best lower bounds (9.58 and 20.29) when only a single type of invariant convex constraint is employed.  As expected, the combination of all three individual constraint sets produces the best overall lower bounds (10.72 and 21.60).  More broadly, these results demonstrate the effectiveness of our approach in producing useful lower bounds for graph edit distance problems arising from real data in a computationally tractable manner.  Specifically, for the Alkane dataset the average lower bound 10.72 is obtained using our convex programming framework and the exact value of the average graph edit distance is $15.3$ (which is obtained via combinatorial approaches).  Our results have more interesting implications for the PAH dataset as it is prohibitively expensive to compute the exact average graph edit distance for this dataset due to the large size of its constituents. In particular, the best-known upper bound on the average graph edit distance of PAH is 29.8 \cite{daller2018approximate}. Our convex relaxation framework produces a lower bound of $21.6$ on the average graph edit distance over all pairs of graphs in PAH, which provides a floor on the possible improvement that one should expect to obtain via better algorithms for computing graph edit distances.


\section{Discussion} \label{Discussion}

In this paper we introduce a framework based on convex graph invariants for obtaining lower bounds on the edit distance between two graphs.  Much of the literature on this topic provides methods for computing upper bounds on the edit distance between two graphs by identifying a feasible sequence of edits to transform one graph to the other.  Our approach is qualitatively different in that it is based on convex relaxation and it leads to guaranteed lower bounds on the edit distance.  Further, our approach can be adapted to the structure underlying the two graphs.  We provide both theoretical and empirical support for our method.

There are a number of potential directions for further investigation arising from our paper.  First, our analysis of the performance of the Schur-Horn relaxation could potentially be tightened in order to obtain sharper conditions for the success of our algorithm.  For example, Corollary \ref{CorollaryMain} only utilizes information about the second most-repeated eigenvalue, and while this provides order-optimal scaling results for families such as triangular graphs, it may be possible to improve our analysis to obtain order-optimal bounds for other families as well.  More broadly, a key step in carrying out a precise theoretical analysis of the power of an invariant convex constraint set is to obtain a full understanding of the facial structure of the set, and it would be of interest to develop such a characterization for a larger suite of invariant convex sets than those presented in this paper.  Finally, a commonly encountered question in many applications is to test whether a given graph is a minor of another graph, and it would be useful to extend the framework described in our paper to address this problem.

\bibliographystyle{plain}
\bibliography{GEDbib}



\end{document}